\documentclass[12pt]{article}

\usepackage{amsmath,amssymb,amsthm}
\usepackage{geometry}
\usepackage{hyperref}
\usepackage[round,authoryear]{natbib}
\usepackage{booktabs}
\usepackage{array}
\usepackage{graphicx}
\usepackage{enumitem}
\usepackage{xcolor}
\usepackage{setspace}
\usepackage{titlesec}
\usepackage{abstract}

\hypersetup{
  colorlinks=true,
  linkcolor=black,
  citecolor=black,
  urlcolor=black
}

\theoremstyle{plain}
\newtheorem{theorem}{Theorem}
\newtheorem{proposition}{Proposition}
\newtheorem{lemma}{Lemma}
\newtheorem{corollary}{Corollary}

\theoremstyle{definition}
\newtheorem{definition}{Definition}
\newtheorem{assumption}{Assumption}

\theoremstyle{remark}
\newtheorem{remark}{Remark}

\newcommand{\E}{\mathbb{E}}
\newcommand{\R}{\mathbb{R}}
\newcommand{\Prob}{\mathbb{P}}
\newcommand{\thetastar}{\theta^{*}}
\newcommand{\shat}{\hat{s}}
\newcommand{\sstar}{s^{*}}
\newcommand{\HD}{H_{C,\mathrm{exp}}^{D}}
\newcommand{\HE}{H_{C,\mathrm{exp}}^{E}}
\newcommand{\DD}{\Delta_{D}}

\title{\textbf{A Coordination Theory of NHS Whistleblowing Failure}}

\author{
  Ari Ercole\\
  \small University of Cambridge, UK\\
  \small \texttt{ae105@cam.ac.uk}
}

\date{}

\begin{document}

\maketitle

\begin{abstract}
\noindent
Persistent whistleblowing failure in hierarchical healthcare organisations is typically attributed to insufficient legal protection for reporters or inadequate managerial incentives to investigate. This diagnosis is argued to be structurally incomplete. A repeated three-player game among a whistleblower, Trust management, and a colleague group is solved via global games techniques to obtain a unique equilibrium, identifying three mechanisms existing policy neither recognises nor addresses. \textit{Collective implication}: corroborating testimony exposes the witness's own conduct to scrutiny, deterring support for reports colleagues privately believe genuine and aligning management and colleagues against validation without coordination; employment protection and legal liability are both shown to leave this channel untouched. The model demonstrates that employment protection can increase reporting while leaving validation rates unchanged. \textit{Active elimination}: colleagues can construct a documented counter-narrative that manipulates management's belief about report credibility. Acting independently, a sufficiently credible exposure environment renders this strictly dominated, a discrete deterrence result with no counterpart for passive silence; acting as an organised coalition, colleagues instead trade that uncertainty for certainty of success, at a concealment risk independent action never bears, so the instrument that deters independent elimination does not deter collusion, which is undermined instead by destabilising it from within and, once suspected, must be investigated, confirmed, and punished promptly: delay is shown to compound the underlying harm, permanently degrade the institutional record, and leave the coalition's incentive to reoffend untouched. A \textit{turnover ratchet}: short management tenure means successive assessors inherit dispositions without the reasoning behind them, so the recorded prior degrades geometrically and only through succession. Three interventions follow, targeted respectively at self-incrimination, independent elimination, and collusion: colleague amnesty, adapted from cartel leniency programmes; mandatory named-trust outcome publication, which repairs the institutional record and reproduces the effect of a public inquiry at administrative cost; and mandatory, prompt investigation of suspected coalitions, prioritised ahead of resolving the underlying report, which repairs the record faster still and forgoes none of its value to delay.


\end{abstract}

\newpage
\tableofcontents
\newpage

\section{Introduction}

NHS whistleblowing failure is one of the most investigated and least resolved problems in health policy. Roughly once a decade for the last thirty years, a public inquiry has found the same pattern: declined warnings, punished reporters, institutional denial. The Francis Inquiry into Mid Staffordshire NHS Foundation Trust catalogued years of declined warnings from nurses, junior doctors, and managers \citep{francis2013}; the Ockenden Review found near-identical patterns in maternity services at Shrewsbury and Telford \citep{ockenden2022}; the Messenger Review identified a systemic culture of speaking-up failure across multiple trusts \citep{messenger2022}. Each inquiry reads, in its particulars, like an update of the one before it. The legislative response, most notably the Public Interest Disclosure Act 1998 and the Freedom to Speak Up programme, has not produced stable improvement: whistleblowers continue to be sanctioned at rates that deter reporting \citep{speak2023}, and the cases reaching public attention now look uncomfortably like those from a decade ago. Thirty years is long enough for a great deal of well-intentioned reform; if the standard account of why whistleblowing fails were correct, that account should by now show visible improvement. It has not.

The prevailing explanation treats this as a two-party problem: whistleblowers need stronger protection, management needs stronger incentives to investigate. This paper argues that the party actually deciding most cases, the colleagues who must decide whether to back a report, ignore it, or work against it, is largely absent from that account, and that this omission is why thirty years of reform has not moved the pattern.

Recent cases make the omitted mechanism concrete. At the Countess of Chester Hospital, the consultants who raised the alarm about elevated neonatal mortality were directed by executives into mediation with the nurse they had reported and made to apologise to her in writing \citep{thirlwall2025}. At the same trust, a tribunal later found the chairman and three directors had conspired, in a covert operation styled Project Countess, to remove a chief executive who had made protected disclosures about board conduct, offering a settlement conditional on their withdrawal and deleting messages in an attempt to escape accountability \citep{gilby2025}. At University Hospitals Birmingham, an independent review following a junior doctor's suicide, which her note attributed to the culture of the organisation, described a climate of fear in which colleagues as well as managers labelled those who raised concerns as difficult or clinically weak \citep{bewick2023}. Retaliatory referral of reporting clinicians to the GMC or NMC is a recurrent enough pattern to have been the subject of a review commissioned by the regulator itself \citep{hooper2015}, and the decade of litigation needed to settle the threshold legal question in \citet{day2017} shows a further margin: the cost of reporting is raised not only by dismissal but by the resources an organisation can deploy against the reporter afterwards.

Three findings follow from taking the colleague layer seriously, developed formally below. Corroborating a report requires giving evidence, and giving evidence exposes the witness's own conduct to scrutiny in a way silence does not; when the underlying problem is systemic, this cost falls on anyone who testifies, so colleagues who privately believe a report end up rationally staying silent, and management and colleagues are aligned against validation without coordinating. This is also why two standard remedies miss: protecting the whistleblower does not change whether colleagues back them, and stronger legal liability for management only bites when a decline is likely to come to light, which is precisely when it is least likely to. Second, colleagues are shown at times to go further, actively building a documented case against the reporter; because this leaves an attributable trail, unlike silence, it can be made strictly irrational at any level of confidence once the exposure environment is credible enough, a discrete result silence never admits. Third, organisations are shown to decay even though every manager involved behaves rationally: short tenure means each report is judged by a successor who inherits a file recording outcomes rather than process, and each ambiguous ``unsubstantiated'' record rationally degrades the next assessor's belief in the organisation's reporting culture, a mechanism that requires turnover and vanishes for anyone who remembers their own decisions.

These findings motivate two policy instruments developed in Section~\ref{sec:policy}: colleague amnesty, adapted from cartel leniency programmes in competition law \citep{motta2003,spagnolo2004}, and mandatory named-trust outcome publication, which repairs the degraded institutional record at its root.

\section{Model}

\subsection{Players, timing, and information}

Three types of player are specified. The \textit{whistleblower} $W$ is a single agent by whom a violation is observed and a decision is made whether to report it. \textit{Trust management} $T$ is the organisational principal by whom reports are received and a decision is made whether to validate or decline them. The \textit{colleague group} $\mathcal{C} = \{c_1, \ldots, c_n\}$ is a set of $n$ agents, with $n$ large, by whom the report is observed and an action chosen from $\{\text{Support}, \text{Defect}, \text{Eliminate}\}$.

The game has infinite horizon, indexed by $t \in \{1, 2, \ldots\}$. In each period, the stage game proceeds as follows.

\begin{enumerate}[label=\arabic*.]
  \item A violation of severity $\theta \sim F(\theta)$ occurs, where $F$ has full support on $\R_{+}$.
  \item $\theta$ is observed by $W$, whose own type $\omega \in \{G, V\}$ (genuine or vexatious) is privately known. $W$ chooses $a_W \in \{\text{Report}, \text{Silent}\}$.
  \item If a report is made, each colleague $c_i$ receives a private signal $\theta_i = \theta + \varepsilon_i$, where $\varepsilon_i \sim U[-\varepsilon, \varepsilon]$ independently, and simultaneously chooses an action $a_i$ from $\{\text{Support}, \text{Defect}, \text{Eliminate}\}$.
  \item The report, the aggregate support rate $s = \frac{1}{n}\sum_i \mathbf{1}[a_i = \text{Support}]$, and the persuasion signal $\sigma$ generated by eliminating colleagues (Section~\ref{sec:elimination}) are observed by $T$. A posterior on $W$'s type is formed and $a_T \in \{\text{Validate}, \text{Decline}\}$ is chosen.
\end{enumerate}

$\theta$ is not directly observed by $T$ before a decision is made. The prior $p = \Prob(\omega = G)$ is known to $T$, and $s$ and $\sigma$ are observed in each period. Since $s$ is monotone in $\theta$ under the threshold strategies characterised below, updating on $s$ is simultaneously updating on the severity of the violation and on $W$'s type; the posterior $p(\shat, \sigma)$ defined below encodes both channels.

Two horizons of discounting are distinguished. Exposure of a declined report, should it occur, is realised after an expected lag of $\ell$ periods, with per-period probability $\lambda \in (0,1)$; the discount factors of management and colleagues are $\delta_T, \delta_C \in (0,1)$ respectively, so exposure liabilities are discounted by $\delta_T^{\ell}$ and $\delta_C^{\ell}$.

\subsection{Payoffs}

\subsubsection*{Whistleblower}

$W$'s stage payoff depends on $a_W$ and $a_T$:
\begin{align}
  u_W = \begin{cases}
    B_W - c_W & \text{if } a_W = \text{Report},\; a_T = \text{Validate} \\
    -P_W - c_W & \text{if } a_W = \text{Report},\; a_T = \text{Decline} \\
    0 & \text{if } a_W = \text{Silent}
  \end{cases}
\end{align}
where $B_W > 0$ is the benefit from a validated report, $c_W > 0$ is the personal cost of reporting, and $P_W \geq 0$ is the punishment inflicted when decline occurs.

\subsubsection*{Colleagues: Support and Defect}

For the Support and Defect actions, colleague $c_i$'s stage payoff is:
\begin{align}
  u_{c_i} = \begin{cases}
    B_C\,\theta - k - \phi H_C & \text{if } a_i = \text{Support},\; a_T = \text{Validate} \\
    -k - R_C & \text{if } a_i = \text{Support},\; a_T = \text{Decline} \\
    0 & \text{if } a_i = \text{Defect},\; a_T = \text{Validate} \\
    b_C - \delta_C^{\ell} \lambda \HD & \text{if } a_i = \text{Defect},\; a_T = \text{Decline}
  \end{cases}
  \label{eq:coll_payoffs}
\end{align}

The interpretation of each element is as follows. $B_C\,\theta > 0$ is the benefit from having supported a vindicated whistleblower, scaling with the severity of the vindicated violation: the professional and moral credit from corroborating a grave concern exceeds that from corroborating a minor one. $k > 0$ is the direct cost of supporting (time, friction, visibility). $\phi \in [0,1]$ is the degree of \textit{collective implication} and $H_C \geq 0$ the associated harm; the product $\phi H_C$ is a \textit{testimony self-incrimination cost}, borne only by supporters, because corroborating the report requires giving evidence, and giving evidence places the witness's own conduct and proximity to the violation under investigative scrutiny in a way that silence does not. $R_C \geq 0$ is the reputational cost of having visibly supported a whistleblower who was then declined. $b_C \geq 0$ is the reward from defection when decline follows (management approval, career advancement). $\HD \geq 0$ is the harm sustained by a colleague identified, in a subsequent discovery of the missed concern, as having stood silently by; it is attached to the Defect--Decline cell because such a discovery presupposes that a concern was in fact missed, and it is discounted by $\delta_C^{\ell}\lambda$. A colleague who supported and was overruled is vindicated rather than harmed by later exposure, so no exposure liability is attached to the Support--Decline cell.

It is convenient to define the net gain from defection conditional on decline:
\begin{align}
  \DD \;\equiv\; k + R_C + b_C - \delta_C^{\ell}\lambda \HD .
  \label{eq:DeltaD}
\end{align}

\subsubsection*{Colleagues: Eliminate}

The Eliminate action is qualitatively distinct from Defect. A documented signal $\sigma_i$ directed at management is produced by a colleague who eliminates: complaints, incident reports, negative testimony, contributions to performance management processes concerning $W$. The aggregate signal is $\sigma = \frac{1}{n_E}\sum_{i: a_i = \mathrm{Elim}} \sigma_i$, where $n_E$ is the number of eliminating colleagues. $\sigma$ is strategically constructed to resemble evidence that $W$ is vexatious.

The payoff to Eliminate is:
\begin{align}
  u_{c_i}(\mathrm{Elim}) = \begin{cases}
    b_C + \beta_i - \delta_C^{\ell} \lambda \HE & \text{if } a_T = \text{Decline} \\
    -\mu_i & \text{if } a_T = \text{Validate}
  \end{cases}
  \label{eq:elim_payoffs}
\end{align}
where $\beta_i \geq 0$ is the private benefit from eliminating over and above ordinary defection (deflected scrutiny, heightened management approval), $\HE$ is the harm sustained upon later discovery of the missed concern by a colleague whose participation is documented, and $\mu_i \geq 0$ is the immediate harm sustained when a validated investigation exposes the fabricated counter-narrative directly. The exposure liability $\delta_C^{\ell}\lambda\HE$ is attached to the Decline branch because that is the branch in which there is a missed concern to discover; in the Validate branch the fabrication is exposed by the investigation itself, immediately and with certainty, which is what $\mu_i$ records.

\subsubsection*{Trust management}

Validation costs $C_T > 0$, yields regulatory and reputational benefit $p(\shat,\sigma)\, R_T$ where $p(\shat,\sigma)$ is $T$'s posterior that the report is genuine, and halts the ongoing harm after the current period. Decline avoids $C_T$ but allows the harm, costing $D_T \geq 0$ per period, to continue indefinitely, and creates an exposure liability. The intertemporal payoffs are:
\begin{align}
  U_T(\text{Validate}) &= -\,C_T + p(\shat, \sigma)\, R_T - D_T \label{eq:uval} \\
  U_T(\text{Decline}) &= -\,\frac{D_T}{1-\delta_T} \;-\; \delta_T^{\ell} \lambda \left(L_T + \Phi\, \Omega_T\right) \label{eq:usup}
\end{align}
where $L_T \geq 0$ is the legal and regulatory liability upon exposure of the decline, $\Omega_T \geq 0$ is the reputational catastrophe cost if a systemic missed concern comes to light, and $\Phi \in [0,1]$ is the \textit{systemicity} of the underlying violation. $\Phi$ is a property of the violation itself and scales the catastrophe term; it is conceptually distinct from the colleague-side implication $\phi \leq \Phi$, which measures the share of investigative scrutiny borne personally by a testifying colleague and which, unlike $\Phi$, can be altered by investigation design (Section~\ref{sec:policy}).

\subsubsection*{Assumptions}

\begin{assumption}\label{ass:params}
(i) $B_W, c_W, B_C, k, C_T, R_T > 0$; $\delta_T, \delta_C, \lambda \in (0,1)$; $\phi, \Phi \in [0,1]$ with $\phi \le \Phi$; $\varepsilon > 0$; $\HE \geq \HD \geq 0$; $\beta_i, \mu_i \geq 0$.
(ii) $\DD > 0$: net of discounted exposure risk, defection is profitable when decline follows.
(iii) $p\,R_T < C_T$: investigation of every report is not individually rational for $T$ at the prior.
(iv) MLRP: the likelihoods $f_G(\shat), f_V(\shat)$ of the observed support rate under genuine and vexatious reports satisfy $f_G/f_V$ increasing in $\shat$, and the likelihoods $g_G(\sigma), g_V(\sigma)$ of the persuasion signal satisfy $g_G/g_V$ decreasing in $\sigma$.
\end{assumption}

\begin{assumption}[Dominance regions]\label{ass:dominance}
There exist severities $0 < \theta_L < \theta_H$ such that: (i) for $\theta \geq \theta_H$ the violation is externally verifiable and is validated by $T$ regardless of $s$ and $\sigma$, and $B_C\,\theta_H > k + \phi H_C$; (ii) for $\theta \leq \theta_L$ the report is dismissed by $T$ regardless of $s$.
\end{assumption}

Assumption~\ref{ass:dominance} is the standard device by which dominance regions are guaranteed \citep{morris1998}: catastrophic violations cannot credibly be declined, and trivial ones are never investigated. Under it, Support is strictly dominant for signals above $\theta_H + \varepsilon$ (validation is certain and the vindication benefit exceeds the costs of supporting), and Defect is strictly dominant over Support for signals below $\theta_L - \varepsilon$ (decline is certain and, by Assumption~\ref{ass:params}(ii), defection then strictly outperforms support).

\section{Analysis: Support and Defect subgame}

The equilibrium of the Support/Defect subgame is first characterised; the extension to the Eliminate action is carried out in Section~\ref{sec:elimination}. Throughout this section, $\sigma = 0$.

\subsection{Trust's investigation condition}

From \eqref{eq:uval} and \eqref{eq:usup}, Decline is preferred by $T$ iff $U_T(\text{Decline}) > U_T(\text{Validate})$, which rearranges to:
\begin{align}
  C_T \;>\; p(\shat)\, R_T \;+\; \delta_T^{\ell} \lambda \left(L_T + \Phi\, \Omega_T\right) \;+\; \frac{\delta_T}{1-\delta_T}\, D_T .
  \label{eq:suppress_cond}
\end{align}
The final term is the present value of the \textit{future} harm that validation would prevent: $D_T/(1-\delta_T) - D_T = \delta_T D_T/(1-\delta_T)$.

\begin{proposition}[Trust's investigation condition]\label{prop:suppress}
  Decline occurs iff \eqref{eq:suppress_cond} holds. The right-hand side is increasing in $p$, $\lambda$, $L_T$, $\Phi$, $\Omega_T$, $D_T$, and $\delta_T$, and decreasing in $\ell$.
\end{proposition}

\begin{proof}
  The inequality is a rearrangement of $U_T(\text{Decline}) > U_T(\text{Validate})$. Monotonicity in $p, \lambda, L_T, \Phi, \Omega_T, D_T$ is immediate. In $\delta_T$: $\partial \mathrm{RHS}/\partial \delta_T = \ell\, \delta_T^{\ell-1}\lambda(L_T + \Phi\Omega_T) + D_T/(1-\delta_T)^2 > 0$. In $\ell$: $\partial \mathrm{RHS}/\partial \ell = \delta_T^{\ell}\ln(\delta_T)\,\lambda(L_T + \Phi\Omega_T) < 0$ since $\ln \delta_T < 0$.
\end{proof}

\begin{remark}\label{rem:tenure}
  The management tenure effect operates through two channels simultaneously, consistent with the organisational economics literature on authority and career concerns \citep{holmstrom1999}. As $\delta_T \to 0$ (short tenure), both the discounted exposure liability $\delta_T^{\ell}\lambda(L_T + \Phi\Omega_T)$ and the future-harm term $\delta_T D_T/(1-\delta_T)$ vanish, and the right-hand side collapses to $p\,R_T$, which is below $C_T$ by Assumption~\ref{ass:params}(iii). Decline is then strictly preferred for any report. As $\delta_T \to 1$ (long tenure), the future-harm term diverges and validation is strictly preferred. Management turnover is therefore not merely correlated with decline but is sufficient for it in the limit.
\end{remark}

\begin{remark}\label{rem:LT_gated}
  $L_T$ enters \eqref{eq:suppress_cond} only through the product $\delta_T^{\ell}\lambda L_T$. The marginal effect of statutory liability on $T$'s decision is therefore proportional to the exposure probability: where enforcement visibility is low, no feasible increase in $L_T$ materially changes $T$'s calculus. This gating is exploited repeatedly below.
\end{remark}

\subsection{$T$'s validation threshold in the support rate}

\begin{lemma}[Validation threshold]\label{lemma:val_prob}
  Under Assumption~\ref{ass:params}(iv) there exists a unique threshold $\sstar \in [0,1]$ such that validation occurs iff $\shat \geq \sstar$, defined implicitly by the equation
  \begin{align}
    p(\sstar)\, R_T + \delta_T^{\ell}\lambda(L_T + \Phi\Omega_T) + \frac{\delta_T}{1-\delta_T}D_T = C_T,
    \label{eq:sstar_def}
  \end{align}
  where $p(\shat) = \dfrac{p\, f_G(\shat)}{p\, f_G(\shat) + (1-p)\, f_V(\shat)}$, with the conventions $\sstar = 0$ (validation for every $\shat$) if the left-hand side of \eqref{eq:sstar_def} exceeds $C_T$ at $\shat = 0$, and decline for every $\shat \in [0,1]$, written $\sstar > 1$, if it is below $C_T$ at $\shat = 1$. $\sstar$ is decreasing in $\lambda$, $L_T$, $\Phi$, $\Omega_T$, $D_T$, $\delta_T$, and $p$.
\end{lemma}

\begin{proof}
  By MLRP, $p(\shat)$ is continuous and strictly increasing in $\shat$, so the left-hand side of \eqref{eq:sstar_def} is continuous and strictly increasing in $\shat$ while $C_T$ is constant; a unique crossing exists whenever the boundary conventions do not bind. Each listed parameter raises the left-hand side pointwise, lowering the crossing.
\end{proof}

\subsection{Colleague coordination: global games solution}

Attention is restricted to threshold strategies: Support is played iff $\theta_i \geq \thetastar$ for a common cutoff $\thetastar$. Given such strategies, the realised support rate at true severity $\theta$ is
\begin{align}
  \shat(\theta, \thetastar) = \frac{\theta + \varepsilon - \thetastar}{2\varepsilon}
\end{align}
for $\thetastar \in [\theta - \varepsilon, \theta + \varepsilon]$, and $0$ or $1$ outside this range.

From \eqref{eq:coll_payoffs}, Support is preferred to Defect by a colleague with signal $\theta_i$ who assigns probability $\pi$ to validation iff
\begin{align}
  \pi\left(B_C\,\theta_i - k - \phi H_C\right) \;+\; (1-\pi)\left(-k - R_C\right) \;>\; (1-\pi)\left(b_C - \delta_C^{\ell}\lambda\HD\right),
\end{align}
which, using \eqref{eq:DeltaD}, rearranges to
\begin{align}
  \pi\left(B_C\,\theta_i - \phi H_C + R_C + b_C - \delta_C^{\ell}\lambda\HD\right) \;>\; \DD .
  \label{eq:coll_ineq}
\end{align}
The left-hand side is strictly increasing in $\theta_i$ for $\pi > 0$, which is the state monotonicity the global games argument requires.

\begin{theorem}[Unique threshold equilibrium]\label{thm:unique}
  Under Assumptions~\ref{ass:params} and \ref{ass:dominance}, and for $\varepsilon$ sufficiently small, the Support/Defect subgame has an essentially unique equilibrium, in threshold strategies, with cutoff
  \begin{align}
    \thetastar \;=\; \frac{1}{B_C}\left[\,k + \phi H_C \;+\; \frac{\sstar}{1-\sstar}\,\DD \,\right]
    \label{eq:theta_star}
  \end{align}
  for interior $\sstar$, and realised support rate $\shat(\theta,\thetastar) = (\theta + \varepsilon - \thetastar)/2\varepsilon$.
\end{theorem}

\begin{proof}
  \textit{Dominance regions.} By Assumption~\ref{ass:dominance}(i), for $\theta_i > \theta_H + \varepsilon$ validation is certain ($\pi = 1$) and \eqref{eq:coll_ineq} reduces to $B_C\theta_i > k + \phi H_C$, which holds; Support is strictly dominant. By Assumption~\ref{ass:dominance}(ii), for $\theta_i < \theta_L - \varepsilon$ decline is certain ($\pi = 0$) and \eqref{eq:coll_ineq} reduces to $0 > \DD$, which fails by Assumption~\ref{ass:params}(ii); Defect is strictly dominant.

  \textit{Strategic complementarity.} A lower cutoff used by others raises $\shat$ at every $\theta$, weakly raising the probability that $\shat \geq \sstar$ and hence $\pi$; the incentive to Support is increasing in others' propensity to Support.

  \textit{Laplacian belief at the cutoff.} With uniform noise and threshold play, the marginal colleague, whose signal equals the cutoff exactly, holds uniform beliefs over the fraction of others whose signals exceed the cutoff \citep{morris1998}: $\shat$ is believed to be distributed $U[0,1]$. Hence the probability assigned to validation at the cutoff is $\pi = \Prob(\shat \geq \sstar) = 1 - \sstar$.

  \textit{Indifference.} Equating the expected payoffs of Support and Defect at $\theta_i = \thetastar$ with $\pi = 1-\sstar$:
  \begin{align*}
    (1-\sstar)\big(B_C\thetastar - k - \phi H_C\big) + \sstar\big({-k}-R_C\big) &= \sstar\big(b_C - \delta_C^{\ell}\lambda\HD\big) \\
    (1-\sstar)\big(B_C\thetastar - k - \phi H_C\big) &= \sstar\big(k + R_C + b_C - \delta_C^{\ell}\lambda\HD\big) \\
    (1-\sstar)\,B_C\thetastar &= (1-\sstar)(k + \phi H_C) + \sstar\,\DD ,
  \end{align*}
  yielding \eqref{eq:theta_star}; the full algebra is set out in Appendix A. By iterated deletion of strictly dominated strategies proceeding inward from the dominance regions, this cutoff is the unique survivor for $\varepsilon$ small \citep{carlsson1993,morris1998}.
\end{proof}

\begin{corollary}[Comparative statics of the cutoff]\label{cor:theta_star_comparative}
  For interior $\sstar$:
  \begin{enumerate}[label=(\roman*)]
    \item $\thetastar$ is increasing in $\phi$, $H_C$, $k$, $R_C$, $b_C$, and $\sstar$, and decreasing in $B_C$.
    \item $\thetastar$ is decreasing in $\lambda$ through two reinforcing channels: directly, $\partial\thetastar/\partial\lambda\big|_{\sstar} = -\frac{1}{B_C}\frac{\sstar}{1-\sstar}\,\delta_C^{\ell}\HD < 0$; and indirectly, since $\sstar$ is decreasing in $\lambda$ (Lemma~\ref{lemma:val_prob}) and $\thetastar$ is increasing in $\sstar$.
    \item $\thetastar$ is decreasing in $\HD$ and increasing in $\ell$ (holding $\sstar$ fixed).
  \end{enumerate}
\end{corollary}

\begin{proof}
  Direct differentiation of \eqref{eq:theta_star}: $\partial\thetastar/\partial\phi = H_C/B_C > 0$; $\partial\thetastar/\partial b_C = \frac{1}{B_C}\frac{\sstar}{1-\sstar} > 0$; $\partial\thetastar/\partial \sstar = \DD / \big(B_C(1-\sstar)^2\big) > 0$ by Assumption~\ref{ass:params}(ii); the $\lambda$, $\HD$, and $\ell$ statics follow from $\DD$'s dependence on $\delta_C^{\ell}\lambda\HD$; the indirect $\lambda$ channel follows by composition with Lemma~\ref{lemma:val_prob}.
\end{proof}

\subsection{Collective implication as the central structural variable}

\begin{theorem}[Limits of bilateral intervention]\label{thm:phi_dominance}
  Fix a violation of severity $\theta < \theta_H$ (below the external-verifiability region of Assumption~\ref{ass:dominance}) and suppose $\sstar$ is interior. Define
  \begin{align}
    \bar{\phi}(b_C) \;\equiv\; \frac{1}{H_C}\left[\,B_C(\theta + \varepsilon) - k - \frac{\sstar}{1-\sstar}\,\DD\,\right].
    \label{eq:phibar}
  \end{align}
  For $\phi > \bar{\phi}(b_C)$, $\shat = 0$ and decline obtains. Moreover:
  \begin{enumerate}[label=(\roman*)]
    \item $\bar{\phi}(b_C)$ is decreasing in $b_C$: management career leverage lowers the level of collective implication at which complete colleague defection sets in.
    \item $\bar{\phi}$ and $\thetastar$ are independent of $P_W$: employment protection for the whistleblower has no effect on the colleague equilibrium.
    \item $L_T$ affects $\bar{\phi}$ only through $\sstar$, and $\partial\sstar/\partial L_T$ is proportional to $\delta_T^{\ell}\lambda$. As $\lambda \to 0$, the influence of $L_T$ on the colleague equilibrium vanishes: no feasible increase in statutory liability restores support where exposure probability is low.
  \end{enumerate}
\end{theorem}

\begin{proof}
  $\shat = 0$ iff $\thetastar \geq \theta + \varepsilon$; substituting \eqref{eq:theta_star} and solving for $\phi$ yields \eqref{eq:phibar}. When $\shat = 0$, $T$'s posterior receives no support signal; by Assumption~\ref{ass:params}(iii) and Lemma~\ref{lemma:val_prob} with $\shat = 0 < \sstar$, decline follows. (i): $b_C$ enters \eqref{eq:phibar} through $\DD$ with a negative sign. (ii): $P_W$ appears in neither \eqref{eq:theta_star} nor \eqref{eq:sstar_def}. (iii): differentiate \eqref{eq:sstar_def} implicitly; $\partial \sstar/\partial L_T = -\,\delta_T^{\ell}\lambda \,/\, \big(p'(\sstar)R_T\big) \to 0$ as $\lambda \to 0$.
\end{proof}

\begin{remark}
  Theorem~\ref{thm:phi_dominance} makes precise in what sense the two standard interventions fail. Employment protection is \textit{structurally} incapable: $P_W$ appears in no expression governing colleague or management behaviour, only in the whistleblower's own participation constraint, so it can convert silence into punished reporting but cannot convert decline into validation. Statutory liability is \textit{gated}: its entire influence is carried by the product $\lambda L_T$, so it is ineffective exactly in the low-visibility environments where decline is most attractive. The residual lever on the colleague side is the pair $(\phi, b_C)$, which no existing intervention addresses.

  It is worth being explicit about which real-world instruments this identifies. The Public Interest Disclosure Act 1998 operates on $P_W$: it reduces the punishment a reporter risks by making certain dismissals automatically unfair, and by construction does not touch $\phi$, $b_C$, or $\gamma$. The Freedom to Speak Up programme, including the network of Guardians it established, operates predominantly on $P_W$ as well, since a Guardian's core function is to support and protect the individual reporter; only the weaker, structurally incomplete channel through $b_C$ noted in Section~\ref{sec:policy} lies outside this. Statutory and common-law legal protections for whistleblowers generally, whatever their precise mechanism, share this property: they are protections extended to $W$, and a protection extended to $W$ is a change to $P_W$. Three decades of the dominant UK policy response have accordingly targeted the one parameter Theorem~\ref{thm:phi_dominance} shows cannot reach the colleague coordination failure.
\end{remark}

\subsection{Whistleblower's participation constraint}

\begin{proposition}[Participation constraint]\label{prop:participation}
  Let $\pi_W$ denote $W$'s equilibrium probability of validation conditional on reporting. A report is made iff
  \begin{align}
    \pi_W > \pi_W^* \equiv \frac{P_W + c_W}{B_W + P_W}.
    \label{eq:participation}
  \end{align}
  $\pi_W^*$ is increasing in $P_W$ and $c_W$ and decreasing in $B_W$. In particular, when decline is certain ($\pi_W = 0$) and $P_W + c_W > 0$, silence is the unique best response.
\end{proposition}

\begin{proof}
  $\E[u_W \mid \text{Report}] = \pi_W(B_W - c_W) + (1-\pi_W)(-P_W - c_W) > 0$ rearranges to \eqref{eq:participation}.
\end{proof}

\begin{remark}\label{rem:silence_vs_suppression}
  When the conditions of Theorem~\ref{thm:phi_dominance} hold, decline upon reporting is certain, and by Proposition~\ref{prop:participation} the equilibrium outcome is silence whenever $P_W + c_W > 0$. Statutory protection that reduces $P_W$ lowers $\pi_W^*$ and can thereby move the outcome from silence to reporting; but since reporting is then met with certain decline, the shift is from the Silence equilibrium to the Decline equilibrium. Reports increase while validations do not. The correct measure of policy effectiveness is therefore not the volume of concerns raised but the volume of validated investigations.
\end{remark}

\subsection{Repeated game: the turnover ratchet}

The dynamics of credibility across periods require a carrier, and rational stationarity rules out the obvious candidate. A single long-lived Bayesian manager learns nothing from its own refusals to investigate: the disposition ``unsubstantiated'' is then an artifact of its own action rather than evidence about the report, and in a stationary game with fixed parameters nothing degrades. The carrier is instead \textit{assessor turnover interacting with a coarse record}. Under the short tenure of Remark~\ref{rem:tenure}, the manager who assesses the report at $t+1$ is typically not the manager who disposed of the report at $t$; what is inherited is the file, and the file records dispositions, not the processes behind them. An unsubstantiated disposition may reflect a genuine investigation that found the report vexatious, or a decline that investigated nothing, and the file does not say which.

Formally, let the organisation's reporting culture be a persistent hidden type $J \in \{H, L\}$, under which reports are genuine at rates $q_H > q_L$ respectively, and let $\iota \in (0,1)$ denote the rate at which, in the assessors' common belief, unsubstantiated dispositions arise from genuine investigation rather than decline. Successive assessors carry a common posterior weight $\rho_t$ on $J = H$ and assess incoming reports at the prior $p_t = \rho_t q_H + (1-\rho_t) q_L$.

\begin{proposition}[Turnover ratchet]\label{prop:ratchet}
  Under assessor turnover with a coarse record:
  \begin{enumerate}[label=(\roman*)]
    \item Each unsubstantiated disposition is genuinely informative to a successor. Its likelihood under culture $J$ is $1 - \iota\, q_J$ (the disposition is avoided only if the report was both investigated and genuine), so the posterior odds on $H$ are multiplied by
    \begin{align}
      \Lambda \;=\; \frac{1 - \iota\, q_H}{\,1 - \iota\, q_L\,} \;<\; 1 .
      \label{eq:prior_suppress}
    \end{align}
    \item Along a path of repeated decline, $\rho_t \to 0$ geometrically at rate $\Lambda$, and $p_t \downarrow q_L$ monotonically with $p_t - q_L = O(\Lambda^t)$.
    \item Define $p^{\min}$ as the prior at which \eqref{eq:sstar_def} holds with equality at $\shat = 1$. If $q_L < p^{\min}$, there exists $t^*$ such that for all $t \geq t^*$ decline is unconditional (the $\sstar > 1$ convention of Lemma~\ref{lemma:val_prob}) and the support equilibrium is unreachable without an external shock to the record.
    \item The ratchet vanishes without turnover. An assessor who made the past decline decisions knows the process behind each disposition; conditional on a known decline, the disposition has probability $1$ under either culture, the likelihood ratio is unity, and no degradation occurs.
    \item If decline is accompanied by a persuasion signal $\sigma > 0$ (Section~\ref{sec:elimination}), the fabricated documents enter the file alongside the disposition, and a successor who cannot identify which documents are strategic reads them as further evidence. Modelling the contamination as an additional factor $\Lambda_\sigma = 1 - \alpha\sigma \in (0,1)$, with $\alpha > 0$ and $\alpha\sigma < 1$, multiplying the per-record likelihood ratio, the posterior odds on $H$ are multiplied by $\Lambda \cdot \Lambda_\sigma < \Lambda$: the record degrades strictly faster under elimination than under passive decline, and the convergence in (ii) accelerates to rate $\Lambda\Lambda_\sigma$.
  \end{enumerate}
\end{proposition}

\begin{proof}
  (i) An unsubstantiated disposition fails to arise only if the report was investigated (probability $\iota$) and found genuine (probability $q_J$): its likelihood is $1 - \iota q_J$, and the odds update is the ratio $\Lambda$, with $\Lambda < 1$ iff $q_H > q_L$. (ii) After $t$ such records the posterior odds are $\Lambda^t$ times the initial odds, which tend to $0$; $p_t = q_L + \rho_t(q_H - q_L)$ is decreasing in $t$ with the stated rate. (iii) The left-hand side of \eqref{eq:sstar_def} at $\shat = 1$ is increasing in the prior; since $p_t \downarrow q_L < p^{\min}$, it falls below $C_T$ in finite time and the boundary convention applies thereafter. (iv) Conditional on the disposition being a decline, its probability is $1$ under either culture; the likelihood ratio is $1$ and the posterior is unchanged. (v) The odds update per period is the product of the disposition ratio $\Lambda$ from (i) and the document ratio $\Lambda_\sigma$; since $\Lambda_\sigma < 1$, the product is strictly below $\Lambda$, and after $t$ periods the odds are $(\Lambda\Lambda_\sigma)^t$ times the initial odds.
\end{proof}

\begin{remark}\label{rem:turnover_carrier}
  Parts (ii) and (iv) together locate the ratchet precisely: it is a succession phenomenon, not an evolutionary one. The state variable is the file; the dynamics are the rational inferences of successive short-tenured assessors from a record that pools non-investigation with investigation. The same low $\delta_T$ that makes decline individually attractive (Remark~\ref{rem:tenure}) is what makes its consequences compound: long-tenured management would neither decline nor, remembering its own choices, inherit a degraded record. Short tenure thus enters the model twice, once as motive and once as memory loss, and the two effects reinforce.
\end{remark}

\begin{remark}\label{rem:freeride}
  $\Lambda$ is decreasing in $\iota$: $\partial\Lambda/\partial\iota = (q_L - q_H)/(1-\iota q_L)^2 < 0$. The record therefore degrades fastest where investigation is generally believed to be credible, because it is the credibility of the wider investigative regime that makes an unsubstantiated disposition informative. Each decline launders itself through that credibility: decline is most corrosive precisely in systems that are otherwise most trusted.
\end{remark}

\begin{remark}\label{rem:reputation}
  A complementary deterrence dynamic operates on the participation margin and requires no turnover at all. Prospective whistleblowers and colleagues learn about a persistent management commitment type from the observed treatment of past reporters, in the manner of reputation models \citep{kreps1982,milgrom1982}: each observed punishment raises the posterior that management is committed to decline, lowering $\pi_W$ and raising the campaign confidence $\gamma$. This channel degrades willingness to report and to support, rather than assessed credibility, and operates alongside the record channel of Proposition~\ref{prop:ratchet}.
\end{remark}

\section{Active elimination}\label{sec:elimination}

Coordinated colleague behaviour against whistleblowers is documented empirically by the mobbing and workplace harassment literature \citep{leymann1996,einarsen2011}, but is not modelled therein as a strategic equilibrium; this section supplies that treatment.

\subsection{The persuasion technology}

The colleague strategy space is now extended to include Eliminate. A documented counter-narrative is contributed to by colleagues who eliminate, constituting a persuasion signal $\sigma$ directed at management. Its effect on $T$'s posterior is
\begin{align}
  p(\shat, \sigma) = \frac{p \cdot f_G(\shat)\, g_G(\sigma)}{p \cdot f_G(\shat)\, g_G(\sigma) + (1-p) \cdot f_V(\shat)\, g_V(\sigma)},
  \label{eq:posterior}
\end{align}
which, by Assumption~\ref{ass:params}(iv), is decreasing in $\sigma$: a stronger elimination signal lowers the posterior that the report is genuine. Elimination succeeds when $p(\shat,\sigma)$ is driven below the level at which \eqref{eq:suppress_cond} holds.

\begin{remark}
  Plausible deniability is conferred on management by the persuasion signal. Initiation of the elimination or explicit endorsement is not required of $T$: once the documented posterior has been moved, decline can be attributed to a credibility assessment of the reporter rather than to the content of the report. The elimination and the decline are formally separate acts by separate parties, which is why coordinated decline of this form is extraordinarily difficult to challenge legally.
\end{remark}

\subsection{The elimination decision}

Colleagues are assumed non-pivotal: with $n$ large, the probability that decline occurs is not materially changed by one colleague's switch between Defect and Eliminate, and is denoted $\gamma$, the commonly held \textit{campaign success probability}, formed from public signals about management commitment, expected participation, the whistleblower's external protection, and the visibility of the environment. $\gamma$ is a belief about an aggregate event and is accordingly common across colleagues; it is not a private signal in the sense of $\theta_i$. What differs across colleagues is not this belief but the payoff parameters $\beta_i$ and $\mu_i$ it is weighed against, so the resulting decision rule is a common $\gamma$ compared with an \textit{individual} threshold, denoted $\gamma_i^*$ below.

\begin{proposition}[Elimination threshold]\label{prop:elim_threshold}
  Comparing \eqref{eq:elim_payoffs} against the Defect rows of \eqref{eq:coll_payoffs}: if
  \begin{align}
    \beta_i \;\leq\; \delta_C^{\ell}\lambda\,\big(\HE - \HD\big),
    \label{eq:elim_dominated}
  \end{align}
  then Eliminate is weakly dominated by Defect for every $\gamma \in [0,1]$, and strictly dominated whenever the inequality is strict or $\mu_i > 0$ and $\gamma < 1$. Otherwise, Eliminate is chosen over Defect by colleague $i$ iff $\gamma \geq \gamma_i^*$, where
  \begin{align}
    \gamma_i^* \;=\; \frac{\mu_i}{\;\beta_i - \delta_C^{\ell}\lambda\big(\HE - \HD\big) + \mu_i\;}.
    \label{eq:gamma_star}
  \end{align}
  $\gamma_i^*$ is increasing in $\mu_i$, $\lambda$, and $\HE - \HD$, and decreasing in $\beta_i$.
\end{proposition}

\begin{proof}
  The expected payoffs conditional on the common decline probability $\gamma$ are
  \begin{align*}
    \E[\mathrm{Elim}] &= \gamma\big(b_C + \beta_i - \delta_C^{\ell}\lambda\HE\big) + (1-\gamma)(-\mu_i), \\
    \E[\mathrm{Defect}] &= \gamma\big(b_C - \delta_C^{\ell}\lambda\HD\big) + (1-\gamma)\cdot 0 .
  \end{align*}
  The difference is
  \begin{align}
    \E[\mathrm{Elim}] - \E[\mathrm{Defect}] \;=\; \gamma\Big[\beta_i - \delta_C^{\ell}\lambda\big(\HE-\HD\big)\Big] \;-\; (1-\gamma)\,\mu_i .
    \label{eq:elim_diff}
  \end{align}
  If the bracket is non-positive, \eqref{eq:elim_diff} is non-positive for all $\gamma$, and strictly negative whenever the bracket is strictly negative, or $\mu_i > 0$ and $\gamma < 1$: weak dominance in general, strict dominance under the stated conditions. If the bracket is positive, \eqref{eq:elim_diff} is increasing in $\gamma$ and crosses zero at \eqref{eq:gamma_star}. Comparative statics: $\partial\gamma_i^*/\partial\lambda = \mu_i\,\delta_C^{\ell}(\HE-\HD)\,/\,\mathrm{den}^2 > 0$; the others are immediate.
\end{proof}

\begin{remark}
  Condition \eqref{eq:elim_dominated} is the sharpest result of the section. Elimination differs from defection in kind, not merely in degree: because the campaign is documentary and attributable, its exposure premium $\HE - \HD$ can be made large enough, or the exposure environment $\lambda$ credible enough, that participation is dominated \textit{regardless of confidence in success}. No analogous kill condition exists for passive defection, which remains rational under some beliefs however large its exposure liability, since silence carries no attributable record. Deterrence of elimination is therefore a discrete, achievable policy objective in a way that deterrence of passive defection is not.
\end{remark}

\subsection{Fragility of elimination under scrutiny}

\begin{remark}[Differential deterrence]\label{prop:fragility}
  An increase in $\lambda$ moves both margins against decline, by combining two results already established. It raises each colleague's threshold $\gamma_i^*$ (Proposition~\ref{prop:elim_threshold}'s comparative static, $\partial\gamma_i^*/\partial\lambda > 0$) and pushes the environment toward the kill condition \eqref{eq:elim_dominated}; and it lowers $\thetastar$ (Corollary~\ref{cor:theta_star_comparative}(ii)), expanding support. The elimination margin, however, admits the discrete regime change of \eqref{eq:elim_dominated}, whereas the defection margin moves only continuously: raising $\lambda$ can switch elimination off entirely but can only ever shrink, never eliminate, passive defection.
\end{remark}

\begin{remark}\label{rem:jamming}
  The elimination ratchet (Proposition~\ref{prop:ratchet}(v)) is robust to assessor sophistication in a way the passive ratchet is not. The passive ratchet is defeated by memory: an assessor who knows a disposition was a decline assigns it likelihood ratio one (Proposition~\ref{prop:ratchet}(iv)). The elimination ratchet survives even full awareness that fabrication is possible, because the successor faces a signal-jamming problem: the strategic documents cannot be separated from the genuine ones, and a correctly calibrated discount still leaves the file informative against the reporter. Only the participants themselves, who know which documents were fabricated, can discount fully; every other assessor, internal or external, is misled to some degree. This is the precise sense in which elimination makes the failure self-documenting, and it is the channel by which a skilled management can deliberately disable the whistleblowing function: each engineered elimination lowers the recorded prior against which the next report is judged, and does so in a way that no successor can fully undo from the file alone.
\end{remark}

\subsection{Collusive elimination}\label{sec:collusion}

The elimination decision analysed so far is decentralised: each colleague independently compares a commonly held belief $\gamma$ about aggregate campaign success against an individual threshold $\gamma_i^*$, with no communication between colleagues and no joint commitment. Some of the most damaging documented episodes do not fit this description. At the Countess of Chester Hospital, the tribunal in \citet{gilby2025} found that the chairman and three directors had jointly planned a covert operation, styled Project Countess, to remove a chief executive who had made protected disclosures, coordinating their actions and later deleting messages and documents in an attempt to escape accountability. This is collusion in the ordinary sense: a small group communicating, sharing information, and committing in advance to act as a bloc, closer to the literature on collusion within hierarchical organisations \citep{tirole1986} than to the decentralised aggregation problem modelled above. This subsection extends the model to admit it.

\begin{definition}[Coalition]\label{def:coalition}
  A coalition is a set of $m \leq n$ colleagues who, prior to the realisation of the stage game, can communicate and commit to a joint course of action: either all $m$ members execute Eliminate, or the coalition does not act.
\end{definition}

\begin{remark}[The same technology, applied to Support, is self-organisation]\label{rem:self_organisation}
  Definition~\ref{def:coalition} does not presuppose an adversarial purpose. The same coordination technology, applied to Support rather than Eliminate, is colleagues jointly agreeing to corroborate together rather than each gambling independently on whether enough others will do the same. What distinguishes the two directions is not the technology but its exposure structure. A coalition organised against $W$ must conceal itself, because its members' actions are only individually defensible if attributable to independent judgement rather than to a plan; concealment is what generates the discovery risk $\rho$ and the aggravated exposure $\Psi_i$ of Proposition~\ref{prop:leak}. A coalition organised in support of $W$ has no comparable need for secrecy: corroborating together in the open is not a liability to be concealed, it is the act itself. Under the amnesty of Proposition~\ref{prop:amnesty}, this asymmetry sharpens further, since amnesty removes the one cost, $\phi H_C$, that would otherwise give supporting colleagues a private reason to prefer acting alone. Colleague amnesty is accordingly better understood as a self-organisation instrument than as a mirror image of the anti-collusion instrument developed later in this section: it does not merely lower an individual threshold, it removes the reason supporting colleagues would ever need to hide from one another, permitting exactly the open coordination that Definition~\ref{def:coalition} shows is available, at no risk, once self-incrimination is off the table.
\end{remark}

\begin{assumption}[Persuasion threshold and discovery timing]\label{ass:persuasion_threshold}
  There exists $n^{\dagger} \leq n$ such that a campaign executed by $n_E \geq n^{\dagger}$ colleagues drives $p(\shat,\sigma)$ below the level at which \eqref{eq:suppress_cond} holds regardless of $\shat$, producing the Elimination outcome of Definition~\ref{def:equilibria}(iii) with certainty; campaigns with $n_E < n^{\dagger}$ do not. Consistent with the documented pattern in which such coalitions are established only by a subsequent investigative or judicial process rather than caught in the act, discovery is assumed to occur, generically, after $T$'s decision has already been made and acted upon: a coalition reaching $n^{\dagger}$ is taken to succeed with certainty, and the risk its members bear is a risk of later discovery, not a risk of failing to secure the decision in the first place.
\end{assumption}

A coalition of size $m \geq n^{\dagger}$ removes the aggregation uncertainty entirely for its members: persuasiveness is no longer a matter of belief about others' independent choices, and, by Assumption~\ref{ass:persuasion_threshold}, the decision it targets is not later reopened by whatever discovery eventually follows. What replaces aggregation uncertainty is a different uncertainty, internal to the conspiracy itself and unresolved for as long as it remains unconfirmed.

\begin{proposition}[Post-decision exposure as a nested participation constraint]\label{prop:leak}
  At any point after the decline has been acted upon, each coalition member privately weighs whether to disclose the conspiracy's existence. Let $B_i', c_i', P_i' \geq 0$ denote, respectively, this member's benefit from a validated disclosure of the conspiracy, personal cost of making it, and punishment if it is not believed. By the argument of Proposition~\ref{prop:participation}, the member discloses iff their belief that such a disclosure would be credited exceeds
  \begin{align}
    \pi_i'^{\,*} = \frac{P_i' + c_i'}{B_i' + P_i'}.
    \label{eq:leak_threshold}
  \end{align}
  Let $\rho \in (0,1)$ denote the population probability that a given member's belief clears this threshold. Members' disclosure decisions are independent, so the coalition's existence is never confirmed, by disclosure or by any other route, with probability $(1-\rho)^m$.
\end{proposition}

\begin{proof}
  The disclosure decision is structurally identical to the whistleblower's participation decision of Section~3.4, with the conspiracy in the role of the violation and the coalition member in the role of $W$; \eqref{eq:leak_threshold} is \eqref{eq:participation} relabelled. Independence of $m$ such decisions gives non-discovery probability $(1-\rho)^m$.
\end{proof}

$\rho$ is left as a primitive population probability rather than derived from a further, fully specified game, exactly as $\gamma$ is elsewhere: Proposition~\ref{prop:leak} establishes the isomorphism to the participation constraint and hence what $\rho$ depends on in principle, without closing the model on a second layer of equilibrium computation.

If eventually confirmed, whether by disclosure or by external investigation, the coalition is exposed as a conspiracy rather than as a set of independent eliminators, and this is treated as strictly worse: coordinated concealment is an aggravating fact in its own right, distinct from and additional to ordinary elimination exposure. Let $\Psi_i > \HE$ denote the harm to member $i$ if the coalition is eventually confirmed. Because success is already certain under Assumption~\ref{ass:persuasion_threshold}, a member's expected payoff from joining a coalition of size $m \geq n^{\dagger}$ depends only on whether that later confirmation occurs:
\begin{align}
  U_i(m) \;=\; (1-\rho)^m\left(b_C + \beta_i\right) \;-\; \left[1-(1-\rho)^m\right]\Psi_i .
  \label{eq:coalition_payoff}
\end{align}

\begin{proposition}[Minimal viable coalitions dominate]\label{prop:minimal_coalition}
  Among coalitions satisfying Assumption~\ref{ass:persuasion_threshold}, $U_i(m)$ is strictly decreasing in $m$ for $m \geq n^{\dagger}$. The minimal viable coalition, $m = n^{\dagger}$, strictly dominates every larger coalition.
\end{proposition}

\begin{proof}
  For $m \geq n^{\dagger}$, persuasiveness is already certain and does not improve with further members, while $(1-\rho)^m$ is strictly decreasing in $m$ for $\rho \in (0,1)$. Differentiating \eqref{eq:coalition_payoff} with respect to $m$ (treated as continuous for this purpose): $\partial U_i/\partial m = (1-\rho)^m\ln(1-\rho)\big[(b_C+\beta_i) + \Psi_i\big] < 0$, since $\ln(1-\rho) < 0$ and the bracketed term is positive.
\end{proof}

\begin{remark}
  A coalition recruits additional members only to clear $n^{\dagger}$, never beyond it: every member added past the minimum increases the probability of eventual discovery without adding to the certainty of success, which is already at its maximum. The small, specifically named core in \citet{gilby2025}, four individuals rather than a broad campaign, is consistent with this: a rational conspiracy is compact by construction, not merely by circumstance.
\end{remark}

\begin{remark}\label{rem:collusion_amnesty}
  Collusion dominates decentralised elimination exactly when decentralised coordination is unreliable ($\gamma$ well below one, as in a workforce too dispersed or too uncertain of management's commitment to coordinate independently) but a compact, communicating core can achieve certainty of success at low $m$: comparing \eqref{eq:coalition_payoff} at $m = n^{\dagger}$ against the decentralised expected payoff underlying \eqref{eq:elim_diff}, a member prefers the minimal coalition whenever the certainty of success it buys outweighs the decentralised route's own uncertainty, net of the extra exposure premium $\Psi_i$ carries relative to the decentralised $\mu_i$. This is a reason to expect collusive elimination specifically in the cases where the decentralised mechanism of Section~\ref{sec:elimination} is weakest, rather than as an alternative confined to separate cases.

  The amnesty provision of Proposition~\ref{prop:amnesty} acquires a second function here. Applied to a conspirator rather than to an ordinary supporting colleague, amnesty lowers the effective $P_i'$ in \eqref{eq:leak_threshold}, lowering $\pi_i'^{\,*}$ and thereby raising $\rho$, the population probability of disclosure. Since $(1-\rho)^m$ is decreasing in $\rho$ for every $m$, this lowers the non-discovery probability of any coalition, and by Proposition~\ref{prop:minimal_coalition} the effect is felt first by the largest, already most fragile coalitions, but a sufficiently large increase in $\rho$ can make even the minimal coalition $n^{\dagger}$ inferior to ordinary Defect, closing off collusive elimination as a rational strategy entirely. Where Proposition~\ref{prop:amnesty} was framed as removing a testimony self-incrimination cost, the identical policy instrument, offered to a conspirator rather than a witness, is a leniency programme in the literal sense of \citet{motta2003} and \citet{spagnolo2004}, rather than merely an analogy to one: it converts a conspiracy's own members into its most likely source of exposure.
\end{remark}

\subsubsection*{What eventual discovery does to the record, and why delay is costly}

The persuasion technology of Section~\ref{sec:elimination} treats an unrevealed $\sigma$ as evidence, at face value, against $W$. A confirmed conspiracy is a different kind of event: it is not merely a discounted or nullified signal, since its very occurrence is informative about why a coalition judged the report worth the risk of concealing it, and, by Assumption~\ref{ass:persuasion_threshold}, that confirmation is generic evidence for the historical record rather than a live reversal of $T$'s already-completed decision.

\begin{assumption}[Conspiracy selection]\label{ass:conspiracy_selection}
  Let $L$ denote the event that a coalition forms and is eventually confirmed, and let $\psi_{\omega} \equiv \Prob(L \mid \omega)$. A confirmed conspiracy is more likely to have been mounted against a genuine report than a vexatious one: $\psi_G > \psi_V$.
\end{assumption}

The asymmetry is not assumed for its own sake: colleagues bear the exposure risk $\Psi_i$ of Proposition~\ref{prop:leak} only when the report is judged worth declining badly enough to justify it, and a report is worth that risk precisely when it is credible and damaging, conditions disproportionately satisfied by genuine reports rather than vexatious ones, which are typically dismissed by the ordinary process without need of an elaborate concealment effort.

Proposition~\ref{prop:ratchet}(iv) showed that a decline known to the deciding assessor as a decline carries likelihood ratio one and does not degrade that assessor's own prior. An investigated and confirmed conspiracy goes further, and does so for the historical record rather than for a decision already taken: it extends an equivalent immunity to every future successor who reads the file, and, because mounting a coalition is costly and disproportionately worth that cost against a genuine report, it is not merely neutral but positively informative for the reporting culture.

\begin{proposition}[Investigated conspiracies repair the record, and delay forgoes the repair]\label{prop:conspiracy_repair}
  Suppose a suspected conspiracy is formally investigated and, if confirmed, the finding that a coalition declined a report through coordinated concealment is entered into the institutional record as such, rather than pooled with ordinary unsubstantiated dispositions. Then:
  \begin{enumerate}[label=(\roman*)]
    \item Such a record carries, for any successor who subsequently reads the file, likelihood ratio
    \begin{align}
      \frac{q_H\,\psi_G + (1-q_H)\,\psi_V}{\,q_L\,\psi_G + (1-q_L)\,\psi_V\,} \;>\; 1
      \label{eq:culture_ratio}
    \end{align}
    on the reporting culture, raising rather than merely failing to lower the recorded prior $p_t$ of Proposition~\ref{prop:ratchet}.
    \item This repair requires no delay relative to the underlying report: unlike the outcome publication of Proposition~\ref{prop:publication}, which depends on the subsequent clinical or regulatory fate of the underlying concern becoming known, the informativeness of a confirmed conspiracy is available as soon as the investigation of the coalition concludes, independent of whether the underlying violation is ever separately substantiated.
    \item The repair is conditional on formal investigation and confirmation, and on that investigation being undertaken promptly. A suspicion that is never formally pursued remains, from a future successor's perspective, indistinguishable from an ordinary unsubstantiated disposition, and is pooled into the coarse record exactly as in Proposition~\ref{prop:ratchet}, forgoing the benefit of (i) and (ii) entirely.
    \item Suppose instead that a suspected coalition is left uninvestigated for $t^*$ periods before confirmation, or is never investigated at all. Three costs compound over the delay, none of which the eventual repair, if it comes, reverses. First, the underlying harm continues to accrue at rate $D_T$ per period for the duration of the delay, a present-value cost of $D_T(1-\delta_T^{t^*})/(1-\delta_T)$ that prompt investigation and validation would have avoided. Second, the record entering investigation at $t^*$ is $p_{t^*}$, already degraded by the factor $(\Lambda\Lambda_\sigma)^{t^*}$ of Proposition~\ref{prop:ratchet}(v) relative to $p_0$; the repair in (i) raises this degraded level, not the level a prompt investigation would have started from, so delay leaves a permanent gap between the record achievable by prompt investigation and that achieved after delay, even conditional on eventual confirmation. Third, because a merely suspected, uninvestigated coalition faces no change to $(b_C, \beta_i)$, nothing in the payoffs deters its members from repeating the same behaviour against a subsequent report during the interval of delay: the incentive structure that produced the first decline remains fully intact for as long as investigation is deferred.
  \end{enumerate}
\end{proposition}

\begin{proof}
  (i) A successor updates beliefs about the culture type $J \in \{H,L\}$ of Proposition~\ref{prop:ratchet}, not about the type $\omega$ of any single report; $\psi_G/\psi_V$ is the correct likelihood ratio for updating belief about \textit{this report}, but the successor observes only a labelled record, not $\omega$ itself. The relevant likelihood is obtained by marginalising $\psi_\omega$ over the culture-specific genuineness rate: $\Prob(L \mid J) = q_J\,\psi_G + (1-q_J)\,\psi_V$. The likelihood ratio on $J$ is therefore \eqref{eq:culture_ratio}, which exceeds one because $(q_H - q_L)(\psi_G - \psi_V) > 0$: both factors are positive, the first by $q_H > q_L$ (Proposition~\ref{prop:ratchet}) and the second by Assumption~\ref{ass:conspiracy_selection}. This raises the recorded prior for any reader of the record, not only for the assessor present at the time. (ii) follows because \eqref{eq:culture_ratio} is a property of the conspiracy's occurrence and requires no subsequent event to be observed, in contrast to the $\kappa$-type audit underlying Proposition~\ref{prop:publication}(i). (iii) follows because an unconfirmed suspicion lacks the positive disposition label that (i) requires; absent formal investigation, a successor cannot distinguish it from the ordinary case of Proposition~\ref{prop:ratchet}, where the file records outcome without process. (iv): the harm accrual is the standard geometric present value of $D_T$ over $t^*$ periods; the record degradation is Proposition~\ref{prop:ratchet}(ii),(v) applied for $t^*$ periods before the repair of (i) is realised, and since that repair acts on whatever prior is current at the time of confirmation, the degradation accumulated before $t^*$ is not clawed back; the incentive-structure claim follows because $U_i(m)$ in \eqref{eq:coalition_payoff} depends on $(b_C,\beta_i,\Psi_i,\rho)$ alone, none of which is altered by suspicion that is not acted upon.
\end{proof}

\begin{remark}
  Proposition~\ref{prop:conspiracy_repair} is the formal basis for a specific policy priority, developed in Section~\ref{sec:policy_collusion}: a suspected coalition should be investigated, confirmed, and its members held to account promptly and as a priority, not folded into, or made to wait on, an investigation of the underlying report. The cost of failing to do so is not merely a foregone opportunity; part (iv) shows it compounds on three separate margins simultaneously, and none of the three is undone by however thorough the eventual investigation turns out to be. The deleted messages and coordinated concealment established by the tribunal in \citet{gilby2025}, established only years after the underlying events, illustrate both the value of eventual confirmation and its limits: the finding of \textit{Project Countess} as an organised operation is, by Assumption~\ref{ass:conspiracy_selection}, itself evidence that Dr Gilby's original disclosures had merit, and the tribunal's institutional finding is precisely the kind of positively labelled record that part (i) identifies as repairing the credibility of future reports at that trust; but nothing in that eventual vindication restores the years over which the record was left degraded and the underlying conduct unaddressed.
\end{remark}

\begin{remark}
  Two scope limitations are noted. First, the ability of a subset of colleagues to identify one another as willing conspirators and to communicate securely is treated as exogenous when it occurs; nothing here endogenises coalition formation itself, which would require a further layer of matching or network structure. Second, \citet{gilby2025} is a case in which the colluding core sits inside Trust management's own hierarchy rather than among the whistleblower's peers, a configuration the present two-player split between colleagues and $T$ does not cleanly represent; the same coalition apparatus plausibly extends to a colluding sub-group of $T$, but doing so properly requires the hierarchical treatment of management flagged in Section~\ref{sec:limitations}, not a relabelling of the results above.
\end{remark}

\section{Equilibrium characterisation}

The full game is characterised by the support cutoff $\thetastar$ together with, for each non-supporting colleague, an elimination threshold $\gamma_i^*$: colleagues with $\theta_i \geq \thetastar$ support; a colleague with $\theta_i < \thetastar$ defects if $\gamma < \gamma_i^*$ (or if \eqref{eq:elim_dominated} holds for them) and eliminates if $\gamma \geq \gamma_i^*$. From here on, wherever the characterisation is stated at the level of the population rather than the individual colleague, parameters $(\beta_i, \mu_i)$ are taken as common across non-supporters, so that $\gamma_i^*$ reduces to a single $\gamma^*$ and \eqref{eq:elim_dominated} to a single condition; this representative-colleague convention is used for all subsequent aggregate statements, while individual heterogeneity in $(\beta_i,\mu_i)$ is retained wherever the individual decision rule itself is at issue, as in Proposition~\ref{prop:elim_threshold} and Remark~\ref{prop:fragility} above.

\begin{definition}[Stage game equilibria]\label{def:equilibria}
  Four equilibrium outcomes are distinguished:
  \begin{enumerate}[label=(\roman*)]
    \item \textit{Support:} a report is made; $\shat \geq \sstar$; validation occurs.
    \item \textit{Decline:} a report is made; $\shat < \sstar$; non-supporters defect ($\gamma < \gamma^*$); decline occurs.
    \item \textit{Elimination:} a report is made; non-supporters eliminate ($\gamma \geq \gamma^*$ and \eqref{eq:elim_dominated} fails); $\sigma > 0$ drives $p(\shat,\sigma)$ below the validation level; decline occurs.
    \item \textit{Silence:} no report is made; \eqref{eq:participation} fails.
  \end{enumerate}
\end{definition}

\begin{theorem}[Existence and conditions]\label{thm:equil}
  Under Assumptions~\ref{ass:params} and~\ref{ass:dominance}:
  \begin{enumerate}[label=(\roman*)]
    \item Support obtains iff $\pi_W \geq \pi_W^*$, $\shat(\theta,\thetastar) \geq \sstar$, and either $\gamma < \gamma^*$, or \eqref{eq:elim_dominated} holds, or the residual signal $\sigma$ is insufficient to reverse \eqref{eq:suppress_cond}.
    \item Decline obtains iff $\pi_W \geq \pi_W^*$, $\shat < \sstar$, and $\gamma < \gamma^*$ or \eqref{eq:elim_dominated} holds.
    \item Elimination obtains iff $\pi_W \geq \pi_W^*$, $\gamma \geq \gamma^*$, \eqref{eq:elim_dominated} fails, and $\sigma$ suffices to hold $p(\shat,\sigma)$ below the validation level.
    \item Silence obtains iff $\pi_W < \pi_W^*$; in particular whenever decline or elimination would follow reporting with certainty and $P_W + c_W > 0$.
    \item Under assessor turnover, repeated decline drives $p_t \downarrow q_L$ (Proposition~\ref{prop:ratchet}), strictly faster under elimination (Proposition~\ref{prop:ratchet}(v)); if $q_L < p^{\min}$, only outcomes (ii)--(iv) are eventually reachable.
  \end{enumerate}
\end{theorem}

\begin{proof}
  Each part assembles the corresponding threshold conditions: Proposition~\ref{prop:participation} for participation, Theorem~\ref{thm:unique} and Lemma~\ref{lemma:val_prob} for the support margin, and Proposition~\ref{prop:elim_threshold} for the elimination margin, which supplies both the interior threshold $\gamma^*$ and the dominance condition \eqref{eq:elim_dominated}; either failing to hold is sufficient, independently, for elimination not to occur, which is why both appear as separate disjuncts in (i) and (ii). \eqref{eq:posterior} together with \eqref{eq:suppress_cond} gives the effect of $\sigma$ on $T$'s decision. Part (v) follows from the two ratchet propositions and the definition of $p^{\min}$.
\end{proof}

The parameters below are not econometrically estimated; data adequate for calibrating them against actual NHS trusts do not exist in usable form. What follows is a plausibility argument rather than a measurement: a parameter configuration is exhibited that is qualitatively consistent with institutional features repeatedly documented in NHS inquiries, short executive tenure, weak external scrutiny, systemic rather than individual violations, and strong hierarchical control over colleagues' careers, and it is shown that this configuration alone is sufficient to generate the failure pattern. Why the NHS specifically is used as the evidentiary setting for this configuration, rather than a claim that its parameters are unusually extreme, is addressed directly in Section~\ref{sec:inquiry_cycle}.

\begin{corollary}[Generic failure under a parameter configuration plausibly consistent with recurrent NHS failures]\label{cor:nhs_failure}
  Suppose, illustratively, that $\delta_T$ is low (short management tenure), $\lambda$ is low (weak enforcement visibility), $\phi$ is high (systemic violations investigated collectively), and $b_C$ is high (hierarchical career control), and suppose additionally that elimination is rewarded ($\beta_i$ large where $b_C$ is large). Then for all but the highest-severity violations the outcome is Decline, Elimination, or Silence.
\end{corollary}

\begin{proof}
  Low $\delta_T$ and low $\lambda$ drive the right-hand side of \eqref{eq:suppress_cond} toward $p\,R_T < C_T$ (Remark~\ref{rem:tenure}, Assumption~\ref{ass:params}(iii)), so $\sstar$ is high or at its boundary. High $\phi$ and $b_C$ raise $\thetastar$ (Corollary~\ref{cor:theta_star_comparative}(i)); for $\phi > \bar{\phi}(b_C)$, $\shat = 0$ (Theorem~\ref{thm:phi_dominance}). Low $\lambda$ lowers $\gamma^*$ \eqref{eq:gamma_star} and moves the environment away from \eqref{eq:elim_dominated}; with $\beta_i$ large, elimination is individually rational for moderate $\gamma$. Certain decline plus $P_W + c_W > 0$ yields silence via Proposition~\ref{prop:participation}. Assumption~\ref{ass:dominance}(i) exempts only $\theta \geq \theta_H$.
\end{proof}

\section{Higher-order beliefs and robustness}

\subsection{Common knowledge discontinuities}

\begin{proposition}[Common knowledge discontinuity]\label{prop:ck}
  \begin{enumerate}[label=(\roman*)]
    \item If the inequality
    \begin{align}
      \delta_C^{\ell}\lambda\HD \;\geq\; k + R_C + b_C + \frac{1-\sstar}{\sstar}\,\big(k + \phi H_C\big)
      \label{eq:ck_support}
    \end{align}
    becomes common knowledge, then $\thetastar \leq 0$ and Support is dominant for all colleagues: the defect equilibrium collapses to support regardless of history.
    \item If the inequality $\delta_C^{\ell}\lambda\big(\HE - \HD\big) \geq \beta_i$ becomes common knowledge, then Eliminate is strictly dominated for all $\gamma$ and the elimination equilibrium is destroyed.
  \end{enumerate}
\end{proposition}

\begin{proof}
  (i) Rearranging \eqref{eq:theta_star}, $\thetastar \leq 0$ iff $k + \phi H_C + \frac{\sstar}{1-\sstar}\DD \leq 0$ iff \eqref{eq:ck_support}. Since $\theta_i \geq 0$ always, every signal exceeds the cutoff and Support is dominant; iterated deletion is immediate. (ii) is condition \eqref{eq:elim_dominated} of Proposition~\ref{prop:elim_threshold}. In both cases, the discontinuity requires common knowledge: under merely private belief in the inequality, the global games threshold logic continues to apply and the cutoffs may remain interior, consistent with \citet{rubinstein1989}.
\end{proof}

\begin{remark}
  The two parts carry different policy weights. Part (i) is demanding: it requires the discounted defector exposure liability to exceed the full defection premium, grossed up by the odds ratio of the validation threshold. Part (ii) is comparatively cheap: only the \textit{difference} in exposure harms must exceed the \textit{incremental} elimination benefit. The realistic sequence of institutional repair is therefore: first destroy the elimination equilibrium (part (ii)), reverting the system to passive decline; then rebuild the support margin through the continuous channels of Corollary~\ref{cor:theta_star_comparative}.
\end{remark}

\subsection{Signalling contamination under elimination}

When elimination occurs, the informational content of $\shat$ is contaminated by the persuasion signal $\sigma$. Both are observed by $T$, but the truthful component cannot be perfectly separated from the strategic component, and by \eqref{eq:posterior}, $p(\shat, \sigma) < p(\shat)$ for any $\shat$ whenever $\sigma > 0$. A rational $T$ aware of the elimination technology would discount $\sigma$; the discount is bounded, however, by $T$'s inability to verify which documents are strategic, and a well-constructed, internally consistent campaign defeats the discount in practice. $T$'s effective validation threshold is consequently higher under elimination than under passive defection, so decline is more likely at any given level of genuine collegial belief in the report.

\subsection{Informational cascades and elimination initiation}

In practice, colleague actions are sequential rather than simultaneous. A signal is provided to all others by the first colleague to file a complaint against $W$: beliefs about management commitment are updated by the response the complaint receives. If a favourable response is made by management, by acknowledging the complaint, initiating a process, or declining to protect $W$, the campaign success probability $\gamma$ is raised for all remaining colleagues, potentially above $\gamma^*$, triggering a cascade of elimination.

Two implications follow. First, the identity of the first eliminator matters: a cascade is far more effectively triggered by a senior, credible colleague whose complaint is taken seriously than by a junior one. Second, management's response to the first complaint is itself strategic: an uncommitted management wishing to appear neutral while benefiting from the campaign can respond procedurally, without explicit endorsement, and thereby signal commitment to observing colleagues. The asymmetry noted for passive defection is amplified here: an early defection or complaint is consistent with any private signal and is therefore safe, whereas early support reveals $\theta_i \geq \thetastar$; downward cascades are correspondingly easier to start than upward ones.

\section{Policy implications}\label{sec:policy}

\subsection{The binding constraints}

By Theorem~\ref{thm:phi_dominance}, the pair $(\phi, b_C)$ is the binding constraint on the support margin, and neither $P_W$ nor $L_T$ reaches it: the former structurally, the latter because it is gated by $\lambda$. By Proposition~\ref{prop:elim_threshold}, the elimination margin is governed by $(\beta_i, \mu_i, \lambda, \HE - \HD)$, none of which is targeted by any existing intervention.

Colleague action against a report is not one phenomenon but two, and which of the two is operating determines which instrument reaches it most directly. Acting \textit{independently} (Section~\ref{sec:elimination}), each colleague weighs an uncertain, commonly held belief $\gamma$ about aggregate campaign success against an individual threshold $\gamma_i^*$: success is probabilistic, exposure on failure is the ordinary elimination harm $\HE$, and the margin is closed by raising the credibility of exposure until Proposition~\ref{prop:elim_threshold}'s kill condition makes participation irrational \textit{regardless of confidence}, a population-wide effect reached through $\lambda$. Acting as an \textit{organised coalition} (Section~\ref{sec:collusion}) trades that uncertainty for certainty of success, obtained by commitment rather than aggregation: the decision it targets is not reopened once made, since discovery, by Assumption~\ref{ass:persuasion_threshold}, is generically a post-decision event. What a coalition's members bear instead is a distinct, aggravated exposure $\Psi_i > \HE$ should the coalition eventually be confirmed, and a discovery probability $\rho$ with no counterpart in the decentralised case. Raising $\lambda$ still contributes to the calculation a coalition made when it chose to conceal, but it does not itself trigger or expedite the confirmation on which every consequence in Proposition~\ref{prop:conspiracy_repair} depends. What closes this margin most directly is raising $\rho$ from within the coalition, and, once a coalition is suspected, investigating and confirming it promptly rather than allowing the costs of Proposition~\ref{prop:conspiracy_repair}(iv) to compound. The interventions below are organised accordingly: colleague amnesty and named-trust publication reach the support margin and the independent-action margin through $\phi$ and $\lambda$ respectively, while a further, distinct instrument (Section~\ref{sec:policy_collusion}) reaches organised collusion most directly, through $\rho$ and prompt investigation, with $\lambda$ continuing to contribute at the margin rather than dominating it.

\subsection{Colleague amnesty: enabling self-organisation}

Section~\ref{sec:collusion} showed that the coordination technology behind an adversarial coalition, the ability of colleagues to communicate and commit to a joint course of action, carries no inherent direction: applied to Support, it is colleagues choosing to corroborate together rather than each privately gambling on the others (Remark~\ref{rem:self_organisation}). The obstacle to organising in the open is not the absence of a coordination device but the presence of $\phi H_C$, the self-incrimination cost that makes even privately sympathetic colleagues reluctant to testify, alone or together. Removing that cost is therefore best read as licensing self-organisation among supporting colleagues, not merely as a subsidy to isolated individuals.

\begin{proposition}[Amnesty removes the self-incrimination channel]\label{prop:amnesty}
  Suppose immunity from investigative consequences is granted to colleagues who voluntarily support a report, so that effective $H_C = 0$ for supporters. Then:
  \begin{enumerate}[label=(\roman*)]
    \item The cutoff falls by exactly $\phi H_C / B_C$:
    \begin{align*}
      \thetastar_{\text{amnesty}} = \frac{1}{B_C}\left[k + \frac{\sstar}{1-\sstar}\DD\right] \;=\; \thetastar - \frac{\phi H_C}{B_C},
    \end{align*}
    a strict reduction for all $\phi > 0$.
    \item The complete-decline condition of Theorem~\ref{thm:phi_dominance} ceases to depend on $\phi$ altogether: collective implication is eliminated as a channel of colleague deterrence, whatever the systemicity $\Phi$ of the underlying violation.
    \item The strategic alignment between management and colleagues is broken: $T$ remains exposed through $\Phi\,\Omega_T$, but the differential testimony cost to colleagues is zero, so the colleague margin is governed by $k$, $R_C$, $b_C$, and $\HD$ alone.
  \end{enumerate}
\end{proposition}

\begin{proof}
  All three parts follow from setting $H_C = 0$ in \eqref{eq:theta_star} and \eqref{eq:phibar}; the reduction in (i) is exact because $\phi H_C$ enters \eqref{eq:theta_star} additively in the numerator.
\end{proof}

\begin{remark}
  A direct precedent is found in competition law leniency programmes \citep{motta2003,spagnolo2004}, which grant immunity to the first cartel member to cooperate with an investigation and are credited with substantially increasing detection by transforming a game of mutual silence into a race to cooperate. The formal structure is the same: cooperation is individually deterred by self-incrimination, and removing that cost reverses the private calculus. No provision of this kind exists anywhere in NHS governance.
\end{remark}

\subsection{Mandatory named-trust outcome publication}

\begin{assumption}[Differential attributability]\label{ass:differential}
  Publication of whistleblowing outcomes raises the exposure harms of documented participants strictly more than those of passive bystanders: the induced increases satisfy $\Delta \HE > \Delta \HD \geq 0$. This reflects that elimination generates an attributable documentary record, whereas silence does not.
\end{assumption}

\begin{proposition}[Named-trust publication]\label{prop:publication}
  Mandatory public reporting of whistleblowing outcomes at the level of named trusts, recording the disposition of each report and its subsequent clinical and regulatory status, has the following effects:
  \begin{enumerate}[label=(\roman*)]
    \item The record is repaired at its informational root. Publication of each disposition together with the subsequent clinical and regulatory status of the underlying concern reveals process: an unsubstantiated disposition followed by confirmed harm is exposed as a decline of a genuine report, and one followed by a clean record is consistent with genuine investigation. The coarseness on which Proposition~\ref{prop:ratchet} depends, the pooling of non-investigation with investigation in the file, is thereby removed, and the mechanism of Proposition~\ref{prop:ratchet}(iv) is extended from the deciding assessor to every reader of the published record.
    \item Effective $\lambda$ is raised: a standing public record lowers the cost of ex-post connection between decline and subsequent harm, so exposure no longer requires a dedicated inquiry. By Lemma~\ref{lemma:val_prob} and Remark~\ref{rem:LT_gated}, this simultaneously lowers $\sstar$ and un-gates $L_T$.
    \item Under Assumption~\ref{ass:differential}, $\gamma^*$ rises and the environment is pushed toward the kill condition \eqref{eq:elim_dominated}; when $\delta_C^{\ell}\lambda_{\mathrm{eff}}\big(\HE - \HD\big) \geq \beta_i$ is reached and rendered publicly observable, Eliminate becomes strictly dominated (Proposition~\ref{prop:ck}(ii)) at administrative rather than inquiry cost.
    \item The deterrent effect is concentrated on elimination rather than passive defection (Remark~\ref{prop:fragility} with Assumption~\ref{ass:differential}).
  \end{enumerate}
\end{proposition}

\begin{proof}
  (i): Proposition~\ref{prop:ratchet}(i) rests on the file pooling non-investigation with investigation; a published record containing subsequent outcomes permits the separation, so successors condition on process rather than disposition, the likelihood ratio of a revealed decline is unity by the argument of Proposition~\ref{prop:ratchet}(iv), and degradation halts; records already accumulated are re-evaluated on the same basis, partially reversing past degradation. (ii): a permanent record reduces the marginal cost of establishing the decline-harm connection, raising the per-period exposure probability perceived by all players; the consequences for $\sstar$ and for the $\lambda L_T$ product follow from Lemma~\ref{lemma:val_prob} and \eqref{eq:suppress_cond}. (iii): by Proposition~\ref{prop:elim_threshold}, $\gamma^*$ is increasing in both $\lambda$ and $\HE - \HD$; Assumption~\ref{ass:differential} and the rise in $\lambda_{\mathrm{eff}}$ move both arguments toward \eqref{eq:elim_dominated}; the common knowledge clause is Proposition~\ref{prop:ck}(ii). (iv): the elimination margin admits the discrete dominance regime while the defection margin moves only continuously (Remark~\ref{prop:fragility}); Assumption~\ref{ass:differential} additionally concentrates the harm increase on $\HE$.
\end{proof}

\begin{remark}
  Named-trust rather than anonymised publication is essential to every part. Under anonymisation, colleagues cannot identify whether their own trust's record is visible, the exposure harms are not raised for identifiable individuals, the prior audit of part (i) is not trust-specific, and the common knowledge event of part (iii) cannot form. Anonymised national statistics, which is what current Freedom to Speak Up reporting provides, produce none of the stated effects.
\end{remark}

\subsection{A distinct policy lever: investigating suspected collusion}\label{sec:policy_collusion}

Section~\ref{sec:policy}'s opening distinction has a direct policy consequence: publication reaches independent elimination by raising $\lambda$ system-wide; a coalition is reached less directly by the same channel, since $\lambda$ enters a coalition's calculation only through the ordinary exposure it was already prepared to accept when it chose to conceal, not through the concealment decision itself. What disciplines a coalition is confirmation, and Proposition~\ref{prop:conspiracy_repair} gives that confirmation two properties that together amount to a policy priority rather than a discretionary option. First, a confirmed coalition is not a closed file but positive evidence for the underlying report, since mounting a coalition is costly and disproportionately worth that cost when the report is genuine. Second, and more consequentially, Proposition~\ref{prop:conspiracy_repair}(iv) shows that leaving a suspected coalition uninvestigated is not a neutral, cost-free deferral: every period of delay adds to the underlying harm, permanently degrades the record the eventual repair will start from, and leaves the coalition's members facing an unchanged incentive to do the same again. A suspected coalition should accordingly be dismantled and its members held to account as a matter of priority, not as a secondary output of resolving the report it targeted.

Four implications follow directly, and constitute a policy recommendation distinct from outcome publication.

\begin{enumerate}[label=(\roman*)]
  \item \textit{Suspected collusion should trigger its own investigation, immediately and in parallel with, not contingent on, any investigation into the underlying report.} Because the informativeness of Proposition~\ref{prop:conspiracy_repair} attaches to the coalition's existence, not to the report's eventual substantiation, an investigation confined to the original concern forgoes the faster channel entirely, and every period it is deferred adds to the compounding cost of Proposition~\ref{prop:conspiracy_repair}(iv).
  \item \textit{The finding must be recorded as a confirmed coalition, not folded into the disposition of the original report.} Proposition~\ref{prop:conspiracy_repair}(iii) shows that an unconfirmed suspicion is indistinguishable, to a future successor, from an ordinary unsubstantiated disposition; the record-repair value exists only if the file positively labels what was found.
  \item \textit{The investigation should be resourced and timed to conclude regardless of the underlying report's outcome, and as quickly as possible.} Tying the two together, the current default, discards exactly the speed advantage that distinguishes this channel from named-trust publication, and Proposition~\ref{prop:conspiracy_repair}(iv) shows that speed here is not merely convenient but determines how much of the harm, the record degradation, and the coalition's ongoing capability are ever actually addressed.
  \item \textit{Confirmed participants should face consequences proportionate to coordinated concealment, not to ordinary non-disclosure.} This is what the aggravated exposure $\Psi_i$ of Proposition~\ref{prop:leak} represents formally; a policy that investigates but does not correspondingly punish leaves the incentive structure of Proposition~\ref{prop:conspiracy_repair}(iv)'s third cost intact even after confirmation.
\end{enumerate}

The Gilby tribunal's finding of \textit{Project Countess} illustrates both halves of this argument. Confirmed years after the fact, it repaired the trust's record retrospectively, exactly as Proposition~\ref{prop:conspiracy_repair}(i) predicts; but the years of delay between the conduct and its confirmation are the years over which Proposition~\ref{prop:conspiracy_repair}(iv)'s three costs were left to compound, unaddressed, which is precisely what this policy is designed to prevent from being the default.

\subsection{Graduated investigation design}

The distinction between $\phi$ and $\Phi$ is operational here. $\Phi$, the systemicity of the underlying violation, is a fact about the world and cannot be legislated away. $\phi$, the share of investigative scrutiny borne personally by a testifying colleague, is a property of investigation design. When an investigation opens by treating the report as implicating the whole team, $\phi$ is driven toward $\Phi$ before any colleague has committed to a position. A graduated design, in which the individual or systemic character of the violation is established first and communicated before colleague testimony is sought, and in which testimony is taken under terms that separate the witness's evidence from the witness's own conduct, holds $\phi$ below $\Phi$ at the moment support decisions are made. By Corollary~\ref{cor:theta_star_comparative}(i) this lowers $\thetastar$ directly. It is an operational change achievable through regulatory guidance rather than statute, and is not currently standard practice in NHS investigations.

\subsection{Independent reporting channels}

Effective $b_C$ is reduced by Freedom to Speak Up Guardians insofar as the link between individual support and management retaliation is broken. A Guardian appointed by and reporting to the Trust board is, however, structurally part of the system in which $b_C$ is high. Effective independence requires external appointment and employment protection by a body outside the Trust.

\subsection{Complementarity of interventions}

\begin{proposition}[Complementarity]\label{prop:complement}
  Write the support outcome as the product of indicators
  \begin{align*}
    O \;=\; \mathbf{1}\big[\pi_W \geq \pi_W^*\big]\cdot \mathbf{1}\big[\shat \geq \sstar\big]\cdot \mathbf{1}\big[\text{elimination fails or is not attempted}\big],
  \end{align*}
  and consider interventions $x_P$ (employment protection, lowering $\pi_W^*$), $x_A$ (amnesty, lowering $\thetastar$), and $x_L$ (publication, raising $\lambda_{\mathrm{eff}}$, lowering $\sstar$ and $\thetastar$, raising $\gamma^*$). Each intervention weakly raises exactly the factor(s) indicated and none lowers any factor. Then $O$ is supermodular in $(x_P, x_A, x_L)$: the marginal effect of any intervention on $O$ is zero whenever any other factor is binding at zero, and positive only when the remaining factors hold. Under a parameter configuration plausibly consistent with recurrent NHS failures (Corollary~\ref{cor:nhs_failure}), at least two of the three factors are simultaneously binding, so no single intervention changes the outcome.
\end{proposition}

\begin{proof}
  A product of indicator factors, each non-decreasing in the intervention vector, is supermodular on the lattice of intervention profiles: raising one factor from 0 to 1 changes the product iff all other factors equal 1. Under Corollary~\ref{cor:nhs_failure}, the support-rate factor is zero (Theorem~\ref{thm:phi_dominance}) and the participation factor is zero (Remark~\ref{rem:silence_vs_suppression}), and where $\gamma \geq \gamma^*$ the elimination factor is zero as well; hence at least two factors bind and any single intervention leaves $O = 0$.
\end{proof}

\begin{remark}
  The complementarity is of the discrete, lattice kind rather than the smooth cross-partial kind: the support equilibrium has several jointly necessary conditions, each addressed by a different instrument, so instruments are useless alone and effective together. This is the precise sense in which three decades of single-instrument reform, dominated by employment protection, could raise reporting volumes without raising validation rates.
\end{remark}

\section{Discussion}

\subsection{Limitations}\label{sec:limitations}

$T$ is treated as a unitary actor within each period, although succession across periods is now modelled explicitly in Proposition~\ref{prop:ratchet}; real Trust management is additionally a hierarchy in which board members, medical directors, and line managers may have different discount factors, different exposures to $\Phi$, and different roles in initiating or condoning elimination campaigns. Different policy implications would follow from the initiation question, whether elimination is management-directed or emerges from within the colleague group, which a hierarchical model of $T$ would clarify. Explicit coordination among colleagues is treated in Section~\ref{sec:collusion} for the case of a colluding coalition drawn from the colleague group; it does not extend to coalitions spanning colleagues and management, which is the configuration in \citet{gilby2025} and which the hierarchical extension of $T$ just noted would be needed to represent properly. The turnover ratchet takes the investigation-rate belief $\iota$ and the culture types as common and fixed; endogenising $\iota$, so that assessors also learn how often their predecessors declined, would be a natural extension and would be expected to dampen but not eliminate the ratchet, since the identification problem of which dispositions were declines remains. The colleague group is treated as homogeneous in payoff structure beyond private signals; in practice, colleagues differ in seniority, employment type, and professional registration exposure, all of which affect $H_C$, $b_C$, $\beta_i$, and $\gamma$ at the individual level. The campaign success probability $\gamma$ is treated as commonly held and exogenous within the period; endogenising its formation from the sequential-move dynamics sketched above is left for future work, as is endogenising the quality of the persuasion signal, and the coalition-formation process of Section~\ref{sec:collusion} is similarly taken as exogenous when it occurs rather than derived from an underlying matching structure. The severity-scaled vindication benefit $B_C\theta$ is one of several channels by which state monotonicity could be introduced; validation probability rising with evidence quality is another, and the qualitative results do not depend on which is chosen.

The model is also silent on the psychology of whistleblowing and elimination. Behaviour beyond the specified payoffs is affected by moral disengagement, in-group loyalty, and motivated reasoning. It is plausible that post-hoc rationalisation is frequently engaged in by participants, making their conduct feel principled rather than strategic; this could be incorporated as an upward bias in the perceived $\gamma$ or a downward bias in the perceived $\HE$.

\subsection{Scope: the coordination gate versus attrition}\label{sec:scope_attrition}

The stage game analysed here treats the colleague coordination decision as effectively one-shot: signals are drawn, actions are chosen, and the case resolves. This matches the empirical pattern of the large majority of concerns, which are decided at first contact by whether a coalition of sufficient size forms at all, and it is the pattern documented at University Hospitals Birmingham \citep{bewick2023} and in the retaliatory-referral evidence \citep{hooper2015}, where the outcome turns on colleague testimony and regulatory credibility respectively, resolved without an extended contest. It does not fit the small number of cases, disproportionately visible because they reach tribunal, in which a report survives the initial coordination stage and the dispute becomes a prolonged contest of resources: \citet{day2017} is a decade-long fight over the threshold legal question of whistleblower status, and \citet{gilby2025} combines an elimination-style campaign (Project Countess) with a multi-year war fought substantially through litigation, settlement conditioned on withdrawal of the disclosures, and evidence that surfaced only under the compulsion of tribunal discovery.

These cases call for a different equilibrium concept, not a parametric extension of the one used here. A natural formalisation treats the post-formation coalition as engaged in a dynamic coordination problem under costly, individually optional exit: colleagues who remain pay a recurring cost of sustained visibility, management pays a recurring cost of continued resistance subject to a rising hazard of a damaging evidentiary shock, and the dispute resolves when either the coalition contracts below the level needed to keep the case alive or management's accumulated cost crosses its own threshold. This is closer to the dynamic global games of \citet{angeletos2007} and the war-of-attrition tradition following \citet{fudenberg1986} than to the static coordination problem of Section~\ref{sec:elimination}, because the relevant state is now a history of survival rather than a single round of private signals, and because the fact that a case has not yet resolved is itself informative and shapes continued participation. The static cutoff $\thetastar$ is the natural determinant of the coalition's size at formation; whether that coalition then survives an extended contest is a question the present model does not answer. This discussion is scoped deliberately: the coordination failure analysed here is judged to be the first-order determinant of outcomes, since it decides whether a case exists at all, and an attrition-stage model is left as a companion extension.

\subsection{Why the NHS: the inquiry-reform cycle}\label{sec:inquiry_cycle}

The choice of the NHS as the setting for this paper is not a claim that its parameters are extreme relative to other health systems. The one candidate for structural extremity checked directly, short executive tenure, does not hold up: NHS trust chief executives average roughly three to four years in post \citep{cepr2020nhsceo}, which is unremarkable by comparison with hospital CEO tenure in the United States (around five years) or Germany (an annual turnover rate near 25\%) \citep{predictors2025hospitalceo}, or Australia, reported worse than both \citep{mathew2024ceoturnover}. If $\delta_T$ is short in the NHS, it is short because hospital administration is short-tenured everywhere, not because the NHS is distinctive.

What the NHS does supply, and what other health systems typically cannot, is a documented record of the failure recurring under repeated, independent, centrally mandated attempts to fix it. This is a consequence of the NHS being a single national system rather than a fragmented multi-payer market: reform in the NHS is legislated or commissioned once, nationally, and its failure or success is subsequently investigated once, nationally, by a single statutory process. A fragmented system does not generate this signature; reform there is piecemeal across competing providers and payers, and no single inquiry can speak to the system as a whole. The NHS's centralisation is therefore what makes it an unusually clean evidentiary setting for a claim about a structural mechanism, largely independent of whether its parameters happen to be extreme.

That evidentiary record has a recognisable shape, repeated across the cases discussed in this paper and in the wider inquiry literature:

\begin{enumerate}[label=\arabic*.]
  \item A scandal becomes public, typically at a single trust.
  \item A public inquiry is commissioned, examining that trust in detail.
  \item The inquiry issues recommendations.
  \item An implementation programme follows, generally national in name.
  \item Conditions visibly improve, at least at the trust investigated.
  \item The same pattern recurs, at a different trust.
\end{enumerate}

Each stage has a counterpart in the model. Stage 1 is the realisation of a high-severity violation under the parameter configuration of Corollary~\ref{cor:nhs_failure}. Stage 2 is a public event that renders the exposure environment credible and common knowledge at the trust in question, which is precisely the mechanism of Proposition~\ref{prop:ck}: an inquiry does not change any underlying parameter, it changes what is commonly known about $\lambda$ and $\HE - \HD$ at that trust, which is sufficient, per Proposition~\ref{prop:ck}(ii), to destroy the elimination equilibrium there. Stages 3 and 4 are, per the preceding remark, interventions concentrated overwhelmingly on $P_W$, which Theorem~\ref{thm:phi_dominance} shows cannot touch the colleague coordination margin; where an implementation programme also raises local scrutiny, it contributes a further, temporary, trust-specific increase in effective $\lambda$. Stage 5, the visible improvement, is genuine and correctly predicted: Proposition~\ref{prop:ck}'s common knowledge discontinuity is a real equilibrium shift, not a placebo, and it should be expected to hold for as long as the inquiry's scrutiny and the elevated $\lambda$ it produced remain salient at that trust.

Stage 6 is where the cycle's own logic explains its own recurrence. The common knowledge event of an inquiry is local: it is common knowledge \textit{at the investigated trust} that the exposure environment has changed, not common knowledge system-wide. Every other trust's $\phi$, $b_C$, and $\delta_T$ are untouched, its turnover ratchet (Proposition~\ref{prop:ratchet}) continues degrading its own institutional record on its own schedule, and nothing in stages 2 through 4 has altered the underlying configuration that made failure generic there in the first place. A cycle that repairs one trust's common knowledge while leaving every other trust's structural parameters and degraded record exactly where they were is a cycle that, by construction, must recur elsewhere. This is the precise sense in which the policy proposals of Section~\ref{sec:policy} are designed to break the cycle rather than extend it: colleague amnesty and named-trust publication are specified as standing, system-wide instruments, not as responses commissioned one trust at a time after the fact, because Proposition~\ref{prop:ck} shows that only a common knowledge event of that scope, rather than a sequence of local ones, reaches every trust the model predicts is otherwise generically at risk.

\subsection{Relation to prior scandals}

The comparisons below are qualitative consistency checks against documented inquiry findings, not parameter estimates: no claim is made that the model's parameters have been measured from data, only that the configuration invoked illustratively in Corollary~\ref{cor:nhs_failure} matches the pattern of conditions each inquiry independently reports.

The configuration with active elimination is consistent with documented conditions at Mid Staffordshire Trust prior to the Francis Inquiry \citep{francis2013}: systemic failure implying high $\Phi$ and, under the investigative practice of the time, high $\phi$; strong management career leverage ($b_C$ high); low regulatory visibility ($\lambda$ low); and rapid management turnover ($\delta_T$ low, Remark~\ref{rem:tenure}). The documented pattern of complaints being filed against staff who raised concerns, and the systematic framing of those concerns as personal grievances, is consistent with the elimination outcome of Theorem~\ref{thm:equil}(iii). The cascade of testimony following the inquiry announcement is consistent with Proposition~\ref{prop:ck}: the inquiry rendered the exposure environment credible and common knowledge, destroying the elimination equilibrium first and then the defection equilibrium.

Documented cases of community nurses and midwives raising concerns about staffing prior to the Ockenden Review \citep{ockenden2022} are consistent with the silence outcome of Theorem~\ref{thm:equil}(iv): high private knowledge of severity, weak capacity to corroborate across colleagues (high $\varepsilon$), and high $P_W$ under precarious employment.

\subsection{Conclusion}

Whistleblowing failure in the NHS is not an anomaly or a compliance gap. It is the equilibrium outcome of a game whose parameters, under a configuration plausibly consistent with recurrent NHS failures rather than one estimated from data, are jointly satisfied by the institutional structure of publicly funded hierarchical healthcare. Three mechanisms are identified that are not recognised in the existing policy literature: testimony self-incrimination as the structural condition aligning colleagues with management against validation; active elimination as a formally distinct failure mode that, uniquely, admits a discrete deterrence condition, and that further separates into an independent form, deterred by raising the credibility of exposure until participation is irrational regardless of confidence, and an organised form, which trades that uncertainty for certainty of success at a concealment risk independent action never bears and which the same instrument therefore cannot reach; and a turnover ratchet in the institutional record, by which successive short-tenured assessors rationally read inherited declines as evidence against report credibility, a degradation that elimination accelerates and that no successor can undo from the file alone, whose repair requires either a costly public inquiry or the low-cost alternative of named-trust outcome publication.

The dominant legislative response, protection of the whistleblower, is necessary but addresses a constraint that binds neither colleagues nor management. Colleague amnesty, named-trust publication, and prompt, prioritised investigation and punishment of suspected coalitions each address a distinct binding constraint identified above, self-incrimination, independent elimination, and organised collusion respectively, and none has precedent in NHS governance. That independent and organised colleague action require different instruments despite presenting as the same surface behaviour, a coordinated campaign against a reporter, is itself a finding of this paper, not merely an implementation detail of it. Their absence from the policy toolkit is not accidental: the game would be changed for the players who currently benefit from the game as it is played.

\appendix
\section{Derivation of the cutoff \eqref{eq:theta_star}}

At the cutoff signal $\theta_i = \thetastar$, the marginal colleague holds Laplacian beliefs over the support rate, $\shat \sim U[0,1]$, so probability $1 - \sstar$ is assigned to validation. Indifference between Support and Defect requires
\begin{align*}
  (1-\sstar)\big(B_C\thetastar - k - \phi H_C\big) + \sstar\big(-k - R_C\big)
  \;=\; (1-\sstar)\cdot 0 + \sstar\big(b_C - \delta_C^{\ell}\lambda\HD\big).
\end{align*}
Collecting terms:
\begin{align*}
  (1-\sstar)\,B_C\thetastar
  &= (1-\sstar)\big(k + \phi H_C\big) + \sstar\big(k + R_C + b_C - \delta_C^{\ell}\lambda\HD\big) - \sstar\, k + \sstar\, k \\
  &= (1-\sstar)\big(k + \phi H_C\big) + \sstar\,\DD ,
\end{align*}
where the second line uses \eqref{eq:DeltaD}. Dividing by $(1-\sstar)B_C$ yields \eqref{eq:theta_star}. Existence of the dominance regions required for the iterated-deletion argument is guaranteed by Assumption~\ref{ass:dominance}, and uniqueness for small $\varepsilon$ follows from the monotonicity of the payoff difference \eqref{eq:coll_ineq} in $\theta_i$ and in others' cutoff, by the argument of \citet{carlsson1993} and \citet{morris1998}. \qed

\section{Proof of Proposition \ref{prop:ck}}

\textit{Part (i).} From \eqref{eq:theta_star}, $\thetastar \leq 0$ iff $(1-\sstar)(k + \phi H_C) + \sstar \DD \leq 0$, i.e.\ iff $\DD \leq -\frac{1-\sstar}{\sstar}(k+\phi H_C)$, i.e.\ iff
\begin{align*}
  \delta_C^{\ell}\lambda\HD \;\geq\; k + R_C + b_C + \frac{1-\sstar}{\sstar}\big(k + \phi H_C\big),
\end{align*}
which is \eqref{eq:ck_support}. Under common knowledge of this inequality, every signal $\theta_i \geq 0 > \thetastar$ lies above the cutoff and Support is strictly dominant; one round of deletion suffices. Note that when \eqref{eq:ck_support} holds, Assumption~\ref{ass:params}(ii) is violated, which is precisely the point: the common knowledge event changes the payoff regime, and the lower dominance region ceases to exist.

\textit{Part (ii).} Condition \eqref{eq:elim_dominated} makes the bracket in \eqref{eq:elim_diff} non-positive, so $\E[\mathrm{Elim}] - \E[\mathrm{Defect}] \leq -(1-\gamma)\mu_i \leq 0$ for all $\gamma \in [0,1]$, with strict inequality whenever $\gamma < 1$ or $\mu_i > 0$. Eliminate is deleted for every colleague. Under merely private belief, a colleague uncertain whether others believe the inequality may still assign $\gamma \geq \gamma_i^*$ computed at their own subjective parameters $(\beta_i, \mu_i)$; the discontinuity attaches to the transition to common knowledge, consistent with \citet{rubinstein1989}. \qed

\bibliographystyle{apalike}
\bibliography{whistleblower}

@book{francis2013,
  author    = {Francis, Robert},
  title     = {Report of the {Mid Staffordshire NHS Foundation Trust Public Inquiry}},
  year      = {2013},
  publisher = {The Stationery Office},
  address   = {London}
}

@book{ockenden2022,
  author    = {Ockenden, Donna},
  title     = {Findings, Conclusions and Essential Actions from the Independent Review of Maternity Services at {The Shrewsbury and Telford Hospital NHS Trust}},
  year      = {2022},
  publisher = {HC 1219, The Stationery Office},
  address   = {London}
}

@book{messenger2022,
  author    = {Messenger, Gordon},
  title     = {Leadership for a Collaborative and Inclusive Future},
  year      = {2022},
  publisher = {Department of Health and Social Care},
  address   = {London}
}

@techreport{speak2023,
  author      = {{NHS England}},
  title       = {Freedom to Speak Up: Annual Report 2022/23},
  year        = {2023},
  institution = {NHS England}
}

@article{carlsson1993,
  author  = {Carlsson, Hans and van Damme, Eric},
  title   = {Global Games and Equilibrium Selection},
  journal = {Econometrica},
  year    = {1993},
  volume  = {61},
  number  = {5},
  pages   = {989--1018}
}

@article{morris1998,
  author  = {Morris, Stephen and Shin, Hyun Song},
  title   = {Unique Equilibrium in a Model of Self-Fulfilling Currency Attacks},
  journal = {American Economic Review},
  year    = {1998},
  volume  = {88},
  number  = {3},
  pages   = {587--597}
}

@article{holmstrom1999,
  author  = {Holmstr{\"o}m, Bengt},
  title   = {Managerial Incentive Problems: A Dynamic Perspective},
  journal = {Review of Economic Studies},
  year    = {1999},
  volume  = {66},
  number  = {1},
  pages   = {169--182}
}

@article{rubinstein1989,
  author  = {Rubinstein, Ariel},
  title   = {The Electronic Mail Game: Strategic Behavior Under ``Almost Common Knowledge''},
  journal = {American Economic Review},
  year    = {1989},
  volume  = {79},
  number  = {3},
  pages   = {385--391}
}

@article{leymann1996,
  author  = {Leymann, Heinz},
  title   = {The Content and Development of Mobbing at Work},
  journal = {European Journal of Work and Organizational Psychology},
  year    = {1996},
  volume  = {5},
  number  = {2},
  pages   = {165--184}
}

@incollection{einarsen2011,
  author    = {Einarsen, St{\aa}le and Hoel, Helge and Zapf, Dieter and Cooper, Cary L.},
  title     = {The Concept of Bullying and Harassment at Work: The {European} Tradition},
  booktitle = {Bullying and Harassment in the Workplace},
  publisher = {CRC Press},
  year      = {2011},
  edition   = {2nd}
}

@article{motta2003,
  author  = {Motta, Massimo and Polo, Michele},
  title   = {Leniency Programs and Cartel Prosecution},
  journal = {International Journal of Industrial Organization},
  year    = {2003},
  volume  = {21},
  number  = {3},
  pages   = {347--379}
}

@article{spagnolo2004,
  author  = {Spagnolo, Giancarlo},
  title   = {Divide et Impera: Optimal Leniency Programs},
  journal = {CEPR Discussion Paper No. 4840},
  year    = {2004}
}

@article{kreps1982,
  author  = {Kreps, David M. and Wilson, Robert},
  title   = {Reputation and Imperfect Information},
  journal = {Journal of Economic Theory},
  year    = {1982},
  volume  = {27},
  number  = {2},
  pages   = {253--279}
}

@article{milgrom1982,
  author  = {Milgrom, Paul and Roberts, John},
  title   = {Predation, Reputation and Entry Deterrence},
  journal = {Journal of Economic Theory},
  year    = {1982},
  volume  = {27},
  number  = {2},
  pages   = {280--312}
}

@misc{thirlwall2025,
  author       = {{Thirlwall Inquiry}},
  title        = {The {Thirlwall Inquiry} into events at the {Countess of Chester Hospital}: public hearings and evidence},
  year         = {2025},
  howpublished = {Statutory public inquiry, hearings 2024--2025}
}

@misc{gilby2025,
  author       = {{Employment Tribunal}},
  title        = {{Dr S Gilby v Countess of Chester Hospital NHS Foundation Trust and I Haythornthwaite, Case Nos 2402398/2023 and 2408654/2023}},
  year         = {2025},
  howpublished = {Liverpool Employment Tribunal, liability judgment 12 February 2025; remedy 2026}
}

@techreport{bewick2023,
  author      = {Bewick, Mike},
  title       = {{University Hospitals Birmingham NHS Foundation Trust}: Independent Rapid Review, Phase 1},
  year        = {2023},
  institution = {IQ4U Consultants}
}

@techreport{hooper2015,
  author      = {Hooper, Anthony},
  title       = {The handling by the General Medical Council of cases involving whistleblowers},
  year        = {2015},
  institution = {General Medical Council}
}

@misc{day2017,
  author       = {{Court of Appeal}},
  title        = {{Day v Health Education England and others [2017] EWCA Civ 329}},
  year         = {2017},
  howpublished = {Court of Appeal (Civil Division)}
}

@article{angeletos2007,
  author  = {Angeletos, George-Marios and Hellwig, Christian and Pavan, Alessandro},
  title   = {Dynamic Global Games of Regime Change: Learning, Multiplicity, and the Timing of Attacks},
  journal = {Econometrica},
  year    = {2007},
  volume  = {75},
  number  = {3},
  pages   = {711--756}
}

@article{fudenberg1986,
  author  = {Fudenberg, Drew and Tirole, Jean},
  title   = {A Theory of Exit in Duopoly},
  journal = {Econometrica},
  year    = {1986},
  volume  = {54},
  number  = {4},
  pages   = {943--960}
}

@article{tirole1986,
  author  = {Tirole, Jean},
  title   = {Hierarchies and Bureaucracies: On the Role of Collusion in Organizations},
  journal = {Journal of Law, Economics, \& Organization},
  year    = {1986},
  volume  = {2},
  number  = {2},
  pages   = {181--214}
}

@misc{cepr2020nhsceo,
  author       = {{CEPR}},
  title        = {The impact of {CEOs} in the public sector: Evidence from the {English} {NHS}},
  year         = {2020},
  howpublished = {CEPR VoxEU column}
}

@article{predictors2025hospitalceo,
  author  = {{Health Care Management Review}},
  title   = {Predictors and effects of hospital chief executive officer turnover},
  journal = {Health Care Management Review},
  year    = {2025}
}

@article{mathew2024ceoturnover,
  author  = {Mathew, Nebu Varughese and Liu, Chaojie and Khalil, Hanan},
  title   = {Causes and Consequences of Health Care {CEO} Turnover in {Australia} and Retention Strategies: A Qualitative Study},
  journal = {SAGE Open Nursing},
  year    = {2024}
}

\end{document}